\documentclass[conference]{IEEEtran}
\usepackage[utf8]{inputenc}

\usepackage{amssymb, amsmath}
\usepackage[margin=1in]{geometry}
\usepackage{amsfonts}
\usepackage{mathtools}
\usepackage{amsthm}
\usepackage{enumerate}
\usepackage{cases}
\usepackage{tikz}
\usepackage{graphicx}
\usepackage{dcolumn}
\usepackage{caption}
\usepackage{subcaption}
\usepackage[normalem]{ulem}
\newcolumntype{2}{D{.}{}{2.0}}

\newtheorem{theorem}{Theorem}
\newtheorem{lemma}{Lemma}

\theoremstyle{definition}
\newtheorem{remark}{Remark}
\newtheorem{model}{Model}

\newcommand{\Ex}{\mathbb{E}}

\DeclarePairedDelimiterX{\norm}[1]{\lVert}{\rVert}{#1}

\title{Bispectrum Unbiasing for Dilation-Invariant
Multi-reference Alignment}

\author{\IEEEauthorblockN{Liping Yin\IEEEauthorrefmark{1},
Anna Little\IEEEauthorrefmark{2}, and Matthew Hirn \IEEEauthorrefmark{3}}
\IEEEauthorblockA{CMSE Department and Mathematics Department, Michigan State University \IEEEauthorrefmark{1}\IEEEauthorrefmark{3}\\
Department of Mathematics and the Utah Center For
Data Science, University of Utah\IEEEauthorrefmark{2}\\
Email: \IEEEauthorrefmark{1}yinlipi1@msu.edu,
\IEEEauthorrefmark{2}little@math.utah.edu,
\IEEEauthorrefmark{3}mhirn@msu.edu}}

\begin{document}

\maketitle

\begin{abstract}
Motivated by modern data applications such as cryo-electron microscopy, the goal of classic multi-reference alignment (MRA) is to recover an unknown signal $f: \mathbb{R} \to \mathbb{R}$ from many observations that have been randomly translated and corrupted by additive noise. We consider a generalization of classic MRA where signals are also corrupted by a random scale change, i.e. dilation.  
We propose a novel data-driven unbiasing procedure which can recover an unbiased estimator of the bispectrum of the unknown signal, given knowledge of the dilation distribution. 
Lastly, we invert the recovered bispectrum to achieve full signal recovery, and validate our methodology on a set of synthetic signals.
\end{abstract}

\section{Introduction to MRA}
In classic multi-reference alignment (MRA), one seeks to recover an unknown signal $f: \mathbb{R} \to \mathbb{R}$ from many observations that have been randomly translated and corrupted by additive noise. 
The MRA problem is a simplified version of problems in cryo-electron microscopy (cryo-EM), and is similar to problems in other fields such as structural biology \cite{diamond1992multiple, park2014assembly, park2011stochastic,  sadler1992shift, scheres2005maximum, theobald2012optimal}, image registration \cite{brown1992survey, foroosh2002extension}, and image processing \cite{zwart2003fast}. 
A formal description of the assumptions is given in Model \ref{model1}.

\begin{model}[Classic MRA] \label{model1}
Suppose we have $M$ independent observations of a real-valued function $f \in L^{2}(\mathbb{R})$ defined by
$$y_j(x) = f(x-t_j) + \varepsilon_j(x), \quad 1 \leq j \leq M,$$
where
\begin{enumerate}[(i)]
\item $\text{supp}(y_j) \subset [-\tfrac{1}{2}, \tfrac{1}{2}]$ for $1 \leq j \leq M$.
\item $\{t_j\}_{j=1}^M$ are independent samples of a random variable $t \in \mathbb{R}$.
\item $\{\varepsilon_j(x)\}_{j=1}^M$ are independent white noise processes on $[-\tfrac{1}{2}, \tfrac{1}{2}]$ with variance $\sigma^2$.
\end{enumerate}
\end{model}

Methods for solving Model \ref{model1} can be grouped into two categories. The first approach is synchronization methods \cite{ bandeira2020non, bandeira2017estimation, bandeira2016low,  bandeira2014multireference,  boumal2016nonconvex,  chen2018projected,  chen2014near, perry2018message, singer2011angular, zhong2018near}, which try to recover each translation factor $\{t_j\}_{j=1}^M$, align the signals using the recovered translation factor, and average the aligned signals to get a smoother estimate for the ground truth signal. However, synchronization methods can be problematic when the signal-to-noise ratio (SNR) is small.
Although synchronization is tenable at small noise levels, 
at high noise levels, the peaks are not recognizable and synchronization fails.
The second approach involves estimating the signal directly, without estimation of the translation factors, using ideas such as the method of moments \cite{hansen1982large, kam1980reconstruction, sharon2019method}. The method of invariants \cite{bandeira2020non, bendory2017bispectrum, collis1998higher, hirn2023power, hirn2021wavelet} is an important subclass of such methods in which one first estimates translation invariant features of the hidden signal (for example its power spectrum and bispectrum), and then an inversion algorithm is applied to recover the hidden signal from the estimated features. Expectation maximization (EM) algorithms \cite{abbe2018multireference, dempster1977maximum} have also shown success for signal recovery, although potential disadvantages include high computation time and convergence to local minima.

Although analysis of Model \ref{model1} has led to important insights \cite{perry2017sample}, it is a highly simplified model which does not include important sources of physical variation in many modern applications. 
For example, in three dimensional applications such as cryo-EM, molecules are randomly rotated and one only has access to 2D projections of the molecule. In addition, objects may undergo changes in size/scale, as well as non-static regions where macro-molecular structures ``flop" around \cite{lim2002modular,ekman2005multi,levitt2009nature,forneris2012modular}. One can think of these deformations as a diffeomorphism acting on the molecule before retrieving the observations. 
In other words, if $g$ is a function extracting the observations and $x$ is a molecule, we retrieve $g(\xi(x))$, where $\xi \in C^2(\mathbb{R}^n)$ has a bijective, continuous hessian. 
However solving Model \ref{model1} in the presence of random diffeomorphisms is highly challenging, and generally must involve either assuming that the diffeomorphism is ``small", or that it has an underlying structure which can be leveraged. In this article, we pursue the latter approach, and specifically focus on the case where the diffeomorphism is a linear function, i.e. the hidden signal is also randomly dilated. As further motivation, we note the analysis of dilations is important in imaging \cite{chandran1992position, capodiferro1987correlation, tsatsanis1990translation, hotta2001scale, martinec2007robust, robinson2007optimal} and audio \cite{omer2013estimation, omer2017time, meynard2018spectral, omer2013estimation, omer2017time, meynard2018spectral} applications.
Incorporating dilations into Model \ref{model1} leads to the formulation of Model \ref{model2} below; see Figure \ref{fig: dilation mra} for an illustration.
\begin{model} \label{model2}
Suppose we have $M$ independent observations of a real-valued function $f \in L^2(\mathbb{R})$ defined by
\begin{align*}
y_j(x) &:= f((1 - \tau_j)^{-1}(x-t_j)) + \varepsilon_j(x) \\
&:= f_j(x) + \varepsilon_j(x), \quad 1 \leq j \leq M.
\end{align*}
Furthermore, assume that
\begin{enumerate}[(i)]
\item $\text{supp}(y_j) \subset [-\tfrac{1}{2}, \tfrac{1}{2}]$ for $1 \leq j \leq M$.
\item $\{t_j\}_{j=1}^M$ are independent samples of a random variable $t \in \mathbb{R}$.
\item $\{\tau_j\}_{j=1}^M$ are independent samples of a uniformly distributed random variable $\tau \in \mathbb{R}$ satisfying $\mathbb{E}[\tau] = 0$ and $\text{Var}(\tau) = \eta^2 \leq \tfrac{1}{12}$.
\item $\{\varepsilon_j(x)\}_{j=1}^M$ are independent white noise processes on $[-\tfrac{1}{2}, \tfrac{1}{2}]$ with variance $\sigma^2$.
\end{enumerate}
\end{model} 
\begin{figure}
    \centering
    \includegraphics[scale = 0.3]{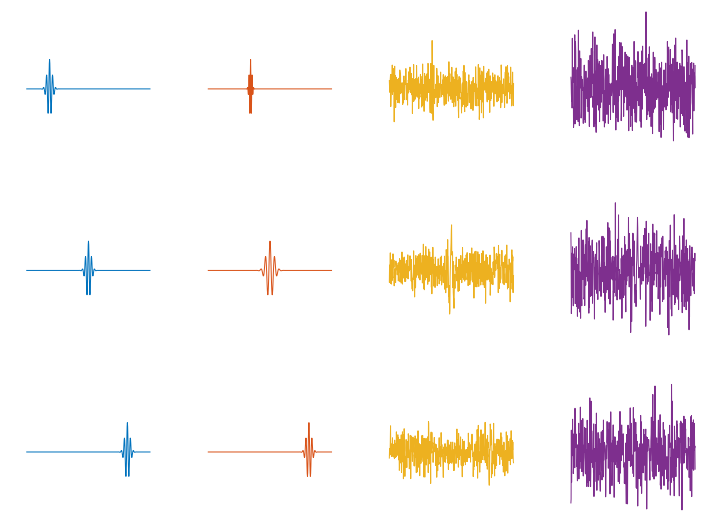}
    \caption{\textbf{Left Column:} three ground truth signals that have been translated without any additive noise. \textbf{Middle Left Column:} Dilating each of the signals. \textbf{Middle Right Column:} Adding Gaussian noise with $\sigma^2 = 0.5$ to each of the signals. \textbf{Right Column:} Adding Gaussian noise with $\sigma^2 = 2$ to each of the signals. Source: \cite{hirn2023power}.} 
    \label{fig: dilation mra}
\end{figure} 
Although Model \ref{model2} corresponds to the simplest possible class of diffeomorphisms other than translations (i.e. linear functions), solving 
Model \ref{model2} is significantly more difficult than solving Model 1, 
as even small dilations can create large perturbations in the frequency domain, especially for high frequency signals \cite{mallat:scattering2012}. 
Our solution to this problem will be a method of invariants which utilizes translation invariant Fourier features such as the power spectrum and bispectrum. 
However, since the Fourier Transform modulus is unstable to small dilations, we must unbias the Fourier invariants to remove the impact of dilations before inversion. 

Model \ref{model2} was also considered in \cite{hirn2021wavelet, hirn2023power}. However both papers only propose estimators for the \textit{power spectrum} of the hidden signal and thus do not fully solve Model \ref{model2}. \cite{hirn2021wavelet} use wavelet-based invariants and a differential unbiasing procedure, 
but the method fails to produce a fully unbiased estimator of the power spectrum. \cite{hirn2023power} develop a novel unbiasing technique which can be applied to the noisy power spectra to obtain a fully unbiased power spectrum estimator. In this article we develop a similar unbiasing procedure that yields an unbiased estimator of the \textit{bispectrum} of the hidden signal. Since the bispectrum is invertible under fairly general conditions, we thus provide the first method capable of achieving full signal recovery for Model \ref{model2}.

For convenience we also define the following 
noiseless model, Model \ref{model3}, given by: 
\begin{model} \label{model3}
Suppose we have $M$ independent observations of a function real-valued $f \in L^2(\mathbb{R})$ defined by
$$f_j(x) = f((1 - \tau_j)^{-1}(x-t_j)), \quad 1 \leq j \leq M$$
Furthermore, assume that
\begin{enumerate}[(i)]
\item $\text{supp}(f_j) \subset [-\tfrac{1}{2}, \tfrac{1}{2}]$ for $1 \leq j \leq M$.
\item $\{t_j\}_{j=1}^M$ are independent samples of a random variable $t \in \mathbb{R}$.
\item $\{\tau_j\}_{j=1}^M$ are independent samples of a uniformly distributed random variable $\tau \in \mathbb{R}$ with $\mathbb{E}[\tau] = 0$ and $\text{Var}(\tau) = \eta^2 \leq \tfrac{1}{12}$.
\end{enumerate}
\end{model}
We note that solving Model \ref{model3} is in fact trivial, since in the absence of additive noise one can first estimate $\|f\|_2$ and then recover $f$ (up to translation) by an appropriate dilation of any observation. However 
to solve Model \ref{model2}, it is convenient to first solve Model \ref{model3} via Fourier invariants and then generalize the approach to the noisy setting.

The remainder of the article is organized as follows. Section \ref{sec:notation} establishes necessary notation and definitions. Section \ref{sec:DilationMRA} gives bispectrum recovery results for Model \ref{model3}. Section \ref{sec:NoisyDilationMRA} gives bispectrum recovery results for Model \ref{model2}. Section \ref{sec:optimization} discusses how the proposed method can be implemented by solving a convex optimization problem. Section \ref{sec:numerics_BSrecovery} reports numerical experiments investigating the accuracy of bispectrum recovery, and Section \ref{sec:numerics_HSrecovery} reports numerical experiments investigating the accuracy of full hidden signal recovery. Section \ref{sec:conclusion} concludes the article.


\section{Notation and preliminaries}
\label{sec:notation}
We let $L^{q}(\mathbb{R})$ represent the set of functions $f$ such that $\|f\|_q^q = \int_{\mathbb{R}} |f|^q \, dx < \infty$. The Fourier Transform of $f \in L^1(\mathbb{R})$ is 
\begin{equation}
    \hat{f}(\omega) = \int_{\mathbb{R}} f(t)e^{-iwt} \,dt,
\end{equation} the power spectrum is
\begin{equation}
(Pf)(\omega) = |\hat{f}(\omega)|^2,
\end{equation} and the bispectrum is
\begin{equation}
Bf(\omega_1, \omega_2) = \hat{f}(\omega_1)\hat{f}^{*}(\omega_2)\hat{f}(\omega_2-\omega_1), ,
\end{equation} 
where $h^*$ denotes the complex conjugate of $h$.
Let $g:= Bf$ where $f$ is the hidden signal. For the second and third models, let
$$g_\eta(\omega_1, \omega_2) := \mathbb{E}_\tau[(Bf_j)(\omega_1, \omega_2)].$$
We will also define the following constants and operations used in the rest of the paper:
\begin{align}
	\label{equ:PSConstants}
	C_0 &:= \frac{(1-\sqrt{3}\eta)}{(1+\sqrt{3}\eta)}\, , \ C_1 := 2\sqrt{3}\eta \, , \ C_2 := \frac{1}{1+\sqrt{3}\eta} \,.
\end{align}
Additionally, define the dilation operator $$L_Cg(\omega_1, \omega_2) := C^4g(C\omega_1, C \omega_2).$$ Note that in polar coordinates $(r, \theta)$, we have $L_Cg(r, \theta) = C^4g(Cr, \theta).$ 
To precisely quantify higher order error terms, we define the following functions:
$$ \overline{(Bf)}^{k}(r, \theta) := r^k\max_{\alpha \in[r/2,2r]}|\partial_{\alpha}^{k}(Bf)(\alpha, \theta)|$$
for integer $k$.
Throughout the paper we write $f \lesssim g$ when $f\leq C g$ for an absolute constant $C$. We write $f \lesssim_\alpha g$ or $f=O_\alpha(g)$ when $f \leq C_\alpha g$ for a constant $C$ depending on $\alpha$.

\section{Bispectrum Recovery for Model 3}
\label{sec:DilationMRA}
We first illustrate how to define an unbiased estimator for the bispectrum under Model \ref{model3}, i.e. in the absence of additive noise.
As in \cite{hirn2023power}, 
it is insightful to first consider the case where we have an infinite number of samples and thus perfect access to $g_\eta$. In this case the bisepctrum of the hidden signal can be perfectly recovered, assuming $\eta$ is known, as described in the following theorem.

\begin{theorem} \label{prop: infinite sample estimator}
Assume that $Bf \in C^1(\mathbb{R}^2)$. Then $Bf$ can be recovered from $g_\eta$, namely:
\begin{small}
$$Bf(r, \theta) = (I-L_{C_0})^{-1}C_1 L_{C_2}\left(4 g_\eta(r , \theta) + r \frac{\partial g_\eta}{\partial r} (r , \theta)\right).$$
\end{small}
\end{theorem}
\begin{proof}
Recall that the Bispectrum is given by 
$$Bf(\omega_1, \omega_2) = \hat{f}(\omega_1)\hat{f^{*}}(\omega_2)\hat{f}(\omega_2-\omega_1).$$
The Fourier Transform of each $f_j$ is $e^{-i \omega t_j}(1-\tau_j)\hat{f}((1-\tau_j)\omega)$, so
$$(Bf_j)(\omega_1, \omega_2) = (1-\tau_j)^3 (Bf)((1-\tau_j) \omega_1, (1-\tau_j) \omega_2).$$

Since $\tau$ has uniform distribution with variance $\eta^2$, the pdf of $\tau$ has form $p_\tau = \frac{1}{2 \sqrt{3} \eta} \chi_{[-\sqrt{3}\eta, \sqrt{3}\eta]}$. Thus:
\begin{align*}
& g_\eta(\omega_1, \omega_2) = \mathbb{E}_\tau[Bf_j(\omega_1, \omega_2)]\\
&\quad =\mathbb{E}_\tau[(1-\tau)^3g((1-\tau)w_1,(1-\tau)w_2)] \\
&\quad = \int (1-\tau)^3 g((1-\tau)w_1,(1-\tau)w_2)  p_\tau(\tau) \, d\tau \, .
\end{align*}
Now we convert to polar coordinates $(r, \theta)$ and let $\tilde{\tau} = (1 - \tau)r$:
\begin{align*}
g_\eta(r, \theta) &= \frac{1}{2 \sqrt{3}\eta}\int_{-\sqrt{3}\eta}^{\sqrt{3}\eta}(1-\tau)^3 g((1-\tau)r, \theta) d\tau \\
&= \frac{1}{2\sqrt{3}\eta r^4}\int_{(1-\sqrt{3}\eta)r}^{(1+\sqrt{3}\eta)r}  \tilde{\tau}^3g(\tilde{\tau},\theta) \, d\tilde{\tau}.
\end{align*}

Let $H$ be the antiderivative in the variable $w$ for the function $h(w, \theta) = w^3g(w, \theta)$. In other words,
$$\frac{\partial H}{\partial w}(w, \theta) = h(w, \theta) =  w^3g(w, \theta).$$ By Fundamental Theorem of Calculus,
$$2 \sqrt{3} \eta r^4 g_\eta(r , \theta) = H( (1 + 3 \sqrt{\eta})r, \theta) - H( (1 - 3 \sqrt{\eta})r, \theta).$$
Now take derivative with respect to $r$ and divide both sides by $r^3$ to get
\begin{align*}
&2 \sqrt{3} \eta \left(4 g_\eta(r , \theta) + r \frac{\partial g_\eta}{\partial r} (r , \theta)\right) \\
&= (1 + 3 \sqrt{\eta})^4 g((1 + 3 \sqrt{\eta})r, \theta) \\
&- (1 - 3 \sqrt{\eta})^4 g((1 - 3 \sqrt{\eta})r, \theta).
\end{align*}
We now apply the dilation operation $L_{C_2}$ to both sides and rewrite the right side in terms of $I$ and $L_{C_0}$ to get 
$$C_1 L_{C_2}\left(4 g_\eta(r , \theta) + r \frac{\partial g_\eta}{\partial r} (r , \theta)\right) = (I-L_{C_0})g(r, \theta).$$
Thus:
$$g(r,\theta) = (I-L_{C_0})^{-1}C_1 L_{C_2}\left(4 g_\eta(r , \theta) + r \frac{\partial g_\eta}{\partial r} (r , \theta)\right).$$
\end{proof} 

In the finite sample case, $g_\eta$ is not known, but under Model \ref{model3} it can be well approximated by:
$$\widetilde{g}_\eta(\omega_1, \omega_2) := \frac{1}{M} \sum_{j=1}^M (Bf_j)(\omega_1, \omega_2).$$
Motivated by Theorem \ref{prop: infinite sample estimator}, we thus define the following estimator: 
\begin{small}
\begin{align}
\label{equ:est_model3}
&(\widetilde{Bf})(r, \theta)  \\
&\quad := (I-L_{C_0})^{-1} C_1 L_{C_2}\left(4\tilde{g}_\eta(r,\theta) + r \frac{\partial\tilde{g}_\eta}{\partial r}(r, \theta)\right)\, . \notag
\end{align}
\end{small}
To show the estimator $\widetilde{Bf}$ has a small error, we will need the following lemma, which is a straight-forward generalization of Lemma 1 in \cite{hirn2023power}.
\begin{lemma} \label{lemma 1}
Assume $Bf \in C^1(\mathbb{R})$ and $\widetilde{Bf}$ is as defined in \eqref{equ:est_model3}. Then
\begin{align*}
\|Bf(r, \theta) - \widetilde{Bf}(r, \theta)\|_2^2 &\lesssim \|g_\eta(r, \theta) - \tilde{g}_\eta(r , \theta)\|_2^2 \\
&+ \left\| r{\partial_r g_\eta} (r , \theta) - r{\partial_r \tilde{g}_\eta} (r , \theta)\right\|_2^2.
\end{align*}
\end{lemma}

This lemma implies that bounding the mean squared error (MSE) of $\widetilde{Bf}$ reduces to bounding the MSE of $\widetilde{g}_\eta$ and $r\partial_r\widetilde{g}_\eta$; since these estimators are unbiased and based off of $M$ samples, standard arguments show the MSE converges to 0 at rate $O_f(\frac{\eta^2}{M})$ as described in the next theorem.   
\begin{theorem}
\label{thm:dilMRA}
Let  $\widetilde{Bf}$ be as in \eqref{equ:est_model3} and assume that $Bf \in C^3(\mathbb{R}^2)$ and $\overline{(Bf)}^{k}\in L^2(\mathbb{R}^2)$ for $k=2,3$. Then:
\begin{align*}
&\mathbb{E}\left[\|Bf - \widetilde{Bf}\|_2^2\right] \\
&\lesssim \frac{\eta^2}{M}\|(Bf)(r, \theta)\|_2^2 + \frac{\eta^2}{M}\|r \partial_r (Bf)(r, \theta)\|_2^2 \\
&+\frac{\eta^2}{M}\|r^2\partial_{rr}(Bf)(r, \theta) \|_2^2+ O_{f}\left(\frac{\eta^4}{M}\right) \, .
\end{align*}
\end{theorem}
\begin{proof}
See Appendix \ref{app:Thm2proof}.
\end{proof}
We remark that when the dilations are large and/or the hidden signal is high frequency, the empirical average $\widetilde{g}_\eta$ is an extremely distorted approximation of $Bf$, but the unbiasing procedure defined in \eqref{equ:est_model3} produces a very accurate estimator $\widetilde{Bf}$; see Figure \ref{bispectrum recovery example}. 

\section{Bispectrum Recovery for Model 2}
\label{sec:NoisyDilationMRA}
Having successfully recovered the bispectrum under Model \ref{model3}, we now extend the approach to the noisy setting, i.e. Model \ref{model2}. The additive noise creates additional difficulties. 
First, if we compute the MSE on $\mathbb{R}^2$ as is done in Theorem \ref{thm:dilMRA}, it will diverge.
Second, we no longer have access to $\tilde{g}_\eta$, but only to
$$\frac{1}{M}\sum_{j=1}^M By_j.$$
Finally, the empirical average of the bispectra is no longer smooth due to the additive noise, and thus cannot be differentiated as done in \eqref{equ:est_model3}.

To deal with the first challenge, we compute the MSE on a finite domain $\Omega$. To deal with the second challenge, we first perform a $\sigma$-based centering to unbias the bispectra for the additive noise. More specifically, we define
\begin{align}
    \label{equ:add_noise_unbias}
    R_\sigma(\omega_1,\omega_2) &:= \tilde{\mu}(\omega_1)h(\omega_1) + \tilde{\mu}^{*}(\omega_2)h(\omega_2) \\
    &\qquad + \tilde{\mu}(\omega_2-\omega_1)h(\omega_2-\omega_1) \, ,\notag
\end{align}
where $h(\omega) = 2\sigma^2 \frac{\sin\left(\frac{1}{2}\omega\right)}{\omega}$ for $\omega \ne 0$, $h(0) = \sigma^2$, and 
$\tilde{\mu}(\omega) = \frac{1}{M} \sum_{j=1}^M\hat{y}_j(\omega)$ approximates $\mu(\omega) = \mathbb{E}_\tau[\hat{f}_j(\omega)]$. We establish this is the correct additive noise unbiasing term in the next lemma.
\begin{lemma} Assume Model \ref{model2} and $R_\sigma$ as in \eqref{equ:add_noise_unbias}. Then:
\begin{align*}
    \Ex_{\tau,\epsilon}[By_j - R_\sigma] &= \Ex_{\tau}[Bf_j] = g_\eta \, .
\end{align*}  
\end{lemma}
\begin{proof}
Since $\varepsilon_j$ is a white noise process on $[-1/2, 1/2]$ with variance $\sigma^2$, we can write $\hat{\varepsilon}_j(\omega) = \int_{-1/2}^{1/2} e^{-i \omega x} dB_x$ as an integral with respect to a Brownian motion, and it is clear that terms involving an odd power of $\hat{\epsilon}_j$ are mean zero. Thus for fixed $\tau_j$:
\begin{align*}
&\mathbb{E}_{\varepsilon}[By_j(\omega_1, \omega_2)] \\
&= Bf_j(\omega_1,\omega_2) + \mathbb{E}_{\varepsilon}[\hat{f}_j(\omega_1)\hat{\varepsilon}_j^{*}(\omega_2)\hat{\varepsilon}_j(\omega_2 - \omega_1) \\
&+  \hat{f}_j^{*}(\omega_2)\hat{\varepsilon}_j(\omega_1)\hat{\varepsilon}_j(\omega_2 - \omega_1) + \hat{\varepsilon}_j(\omega_1) \hat{\varepsilon}_j^{*}(\omega_2)\hat{f}_j(\omega_2 - \omega_1)] \\
&= Bf_j(\omega_1,\omega_2) + \hat{f}_j(\omega_1)h(\omega_1)  \\
&+ \hat{f}_j^{*}(\omega_2)h(\omega_2) + \hat{f}_j(\omega_2-\omega_1)h(\omega_2-\omega_1)
\end{align*}
since
$\mathbb{E}_{\varepsilon}[\hat{\varepsilon}(\omega_1) \hat{\varepsilon}^{*}(\omega_2)] = 
h(\omega_2-\omega_1)$
(see Theorem 4.5 of \cite{klebaner2012introduction}).
Taking expectation over $\tau$ now gives:
\begin{align*}
  \Ex_{\tau,\epsilon}[By_j] &= g_\eta + \mu(\omega_1)h(\omega_1) + \mu^{*}(\omega_2)h(\omega_2) \\
    &\qquad + \mu(\omega_2-\omega_1)h(\omega_2-\omega_1) \\
    &= g_\eta + \Ex_{\tau,\epsilon}[\tilde{\mu}(\omega_1)h(\omega_1) + \tilde{\mu}^{*}(\omega_2)h(\omega_2) \\
    &\qquad + \tilde{\mu}(\omega_2-\omega_1)h(\omega_2-\omega_1)] \\
    &= g_\eta + \Ex_{\tau,\epsilon}[R_\sigma] \, .
\end{align*}
\end{proof}
After empirical centering by $R_\sigma$, we can decompose the computable quantity into two pieces for easier analysis: 
$$\frac{1}{M}\sum_{j=1}^M By_j - R_\sigma= \tilde{g}_\eta + \tilde{g}_\sigma,$$
where 
\begin{align*}
\tilde{g}_\eta &:=\frac{1}{M}\sum_{j=1}^M Bf_j  \ ,\  \tilde{g}_\sigma := \frac{1}{M}\sum_{j=1}^M (By_j-Bf_j) -R_{\sigma}\, .
\end{align*}
Such a decomposition will allow us to leverage the results for Model \ref{model3} to control $\tilde{g}_\eta$.
We let $R_j := By_j-Bf_j$ denote the deviation in the definition of $\tilde{g}_\sigma$.

Note $\tilde{g}_\eta+\tilde{g}_\sigma$ is not smooth due to the additive noise. Thus to deal with the third challenge, we need to add a smoothing procedure. Let $\phi_L(r) = (2\pi L^2)^{-1/2} e^{-r^2/(2L^2)}$ be a low pass filter. We define a new estimator for Model \ref{model2} as
\begin{align} 
\label{Model 2 estimator}
&(\widetilde{Bf})(r, \theta) := (I-L_{C_0})^{-1}C_1L_{C_2} \\
&\ \ \left[4 ((\tilde{g}_\eta + \tilde{g}_\sigma) \ast \phi_L)(r , \theta) + r\partial_r(( \tilde{g}_\eta + \tilde{g}_\sigma) \ast \phi_L) (r , \theta)\right] \notag.    
\end{align}
To bound the MSE of the above estimator, we will use the following four lemmas. Lemmas \ref{lem:conv_decay} and \ref{lem:conv_decay} generalize Lemmas 2 and 3 in \cite{hirn2023power}; Lemma \ref{lem: noise bound} generalizes Lemma D.1 in \cite{hirn2021wavelet}; the proofs of all three are provided in the supplemental material for completeness. 
\begin{lemma}
\label{lem:conv_decay}
Let $q \in L^2(\mathbb{R}^2)$ and assume $|\hat{q}(\omega)|$ decays like $|\omega|^{-\alpha}$ for some integer $\alpha \geq 2$, i.e. there exist constants $C,\omega_0$ such that $|\hat{q}(\omega)|\leq C|\omega|^{-\alpha}$ for $|\omega|>\omega_0$. Then for $L$ small enough,
$$\|q - q \ast \phi_L\|_2^2 \lesssim \|q\|_2^2 L^4 + L^{4 \land (2\alpha-2)}.$$
\end{lemma}

\begin{lemma}
\label{lem:conv_decay2}
Let $rq(r, \theta) \in L^2(\mathbb{R}^2, dr \times d \theta)$ and assume its Fourier transform $\widehat{(\cdot)q(\cdot)}(\omega)$ decays like $|w|^{-\alpha}$ for some integer $\alpha \geq 2$. Then for $L$ small enough,

\begin{align*}
\|r(q - q \ast \phi_L)(r, \theta)\|_2^2 &\lesssim \|r q(r, \theta)\|_2^2 L^4 + L^{4 \land (2\alpha-2)} \\
&+ (L^3 \|q\|_2^2) \land (L^4 \|\partial_r q\|_2^2).
\end{align*}
\end{lemma}


\begin{lemma} \label{lem: noise bound}
Suppose $\varepsilon$ is a mean zero Gaussian white noise supported on $[-1/2, 1/2]$ with variance $\sigma^2$. For all $p > 0$ and $\omega \in \mathbb{R}$,
$$\mathbb{E}\left[|\hat{\varepsilon}(\omega)|^p\right] \lesssim_p \sigma^p.$$
\end{lemma}

\begin{lemma}
\label{lem:noise_bound}
Assume that the assumptions of Model \ref{model2} hold, $Bf \in C^3(\mathbb{R}^2)$, $(\cdot)\widehat{Bf}(\cdot)$ decays like $|\cdot|^\kappa$ for some $\kappa \geq 3$, and 
and $\overline{(Bf)}^{k}\in L^2(\mathbb{R}^2)$ for $k=2,3$. Then:
$$\|\tilde{g}_\sigma\|_{L^2(\Omega)}^2\lesssim_{\Omega} \frac{\sigma^2}{M} \|f\|_2^4 + \frac{\sigma^4}{M} \|f\|_2^2 +  \frac{\sigma^6}{M}.$$
\end{lemma}
\begin{proof}
See Appendix B.
\end{proof}

We can now bound the error of the bispectrum estimator \eqref{Model 2 estimator} for the noisy dilation MRA model.

\begin{theorem}
\label{thm:main}
Assume that the assumptions of Model \ref{model2} hold, $Bf \in C^3(\mathbb{R}^2)$, $(\cdot)\widehat{Bf}(\cdot)$ decays like $|\cdot|^\kappa$ for some $\kappa \geq 3$, and $\overline{(Bf)}^{k}\in L^2(\mathbb{R}^2)$ for $k=2,3$. For the estimator $(\widetilde{Bf})(r, \theta)$ defined in \eqref{Model 2 estimator},
$$\mathbb{E}\left[\|Bf - \widetilde{Bf}\|^2_{L^2(\Omega)}\right] \leq C_{f,\Omega} \left( \frac{\eta^2}{M} + L^4 + \frac{\sigma^2 \vee \sigma^6}{L^2M} \right),$$
where $C_{f, \Omega}$ only depends on $f$ and $\Omega$.
\end{theorem}
\begin{proof}
First, since $(I-L_{C_0})^{-1}C_1L_{C_2}$ is a bounded linear operator as in Lemma \ref{lemma 1}, 
\begin{align*}
&\|Bf-\widetilde{Bf}\|_{L^2(\Omega)}^2 \\
& \lesssim \left\|g_\eta - \tilde{g}_\eta \right\|_{L^2(\Omega)}^2 +  \left\| r \partial_r g_\eta - r \partial_r \tilde g_\eta \right\|_{L^2(\Omega)}^2 \\
&+\|\tilde{g}_\eta - \tilde{g}_\eta \ast \phi_L\|_{L^2(\Omega)}^2 + \left\| r \partial_r \tilde g_\eta - r \partial_r (\tilde{g}_\eta \ast \phi_L) \right\|_{L^2(\Omega)}^2 \\
&+ \|\tilde{g}_\sigma \ast \phi_L\|_{L^2(\Omega)}^2 +\|r \partial_r(\tilde{g}_\sigma \ast \phi_L)\|_{L^2(\Omega)}^2.
\end{align*}
By Theorem \ref{thm:dilMRA},
\begin{align*}
&\mathbb{E}\left[\left\|g_\eta - \tilde{g}_\eta \right\|_{L^2(\Omega)}^2 +  \left\| r \partial_r g_\eta - r \partial_r \tilde g_\eta \right\|_{L^2(\Omega)}^2 \right] \\ &\lesssim \frac{\eta^2}{M} \|(Bf)(r, \theta)\|_2^2 + \frac{\eta^2}{M} \|r\partial_r(Bf)(r, \theta)\|_2^2 \\
&+ \frac{\eta^2}{M} \|r^2\partial_{rr}(Bf)(r, \theta)]^2\|_2^2 + O_f\left(\frac{\eta^4}{M}\right) \, .
\end{align*}
It remains to bound the other four terms. By Lemma 2, 
$$\|\tilde{g}_\eta - \tilde{g}_\eta \ast \phi_L \|_{L^2(\Omega)}^2 \lesssim \|\tilde{g}_\eta\|_2^2 L^4 + L^{4 \land (2\kappa-2)}.$$
Since
\begin{align*}
\|Bf_j\|_2^2 
&= (1-\tau_j)^4 \|Bf\|_2^2 
\leq 2^4 \|Bf\|_2^2 \, ,
\end{align*} applying triangle inequality yields 
$$\|\tilde{g}_\eta\|_2 = \left\|\frac{1}{M}\sum_{j=1}^M Bf_j\right\|_2 \leq \frac{1}{M}\sum_{j=1}^M \|Bf_j\|_2 \lesssim \|Bf\|_2 \, ,$$
and we obtain
$$\|\tilde{g}_\eta - \tilde{g}_\eta \ast \phi_L \|_{L^2(\Omega)} \lesssim  \|Bf\|_2^2 L^4 + L^{4 \land (2\kappa-2)}.$$
Also, by Lemma \ref{lem:conv_decay2}, 
\begin{align*}
&\|r \partial_r(\tilde{g}_\eta -  (\tilde{g}_\eta \ast \phi_L))\|_{L^2(\Omega)}^2 \\
\quad &\lesssim \|r \partial_r \tilde{g}_\eta\|_2^2 L^4 + L^{4 \land (2\kappa-2)} + L^4 \| \partial_{rr}\tilde{g}_\eta\|_2^2 \\
\quad &\lesssim (\| r\partial_r (Bf)(r, \theta)\|_2^2 + 1+\| \partial_{rr} (Bf)(r, \theta)\|_2^2) L^4 \, .
\end{align*}
Now consider the additive noise terms. 
By Young's Convolution Inequality and Lemma \ref{lem:noise_bound},
\begin{align*}
\|\tilde{g}_\sigma \ast \phi_L\|_{L^2(\Omega)}^2 &\leq \|\phi_L\|_{L^1(\Omega)}^2 \|\tilde{g}_\sigma\|_{L^2(\Omega)}^2\\
&\lesssim \|\tilde{g}_\sigma\|_{L^2(\Omega)}^2\\
&\lesssim_{\Omega} \frac{\sigma^2}{M} \|f\|_2^4 + \frac{\sigma^4}{M} \|f\|_2^2 +  \frac{\sigma^6}{M}\\
&\lesssim_{\Omega, f} \left(\frac{\sigma^2 \lor \sigma^6}{M}\right).
\end{align*}

Now we consider $\|r\partial_r( \tilde g_\sigma \ast \phi_L)\|_{L^2(\Omega)}$. Let $R_\Omega = \max_{\Omega} |r|$.  We have 
\begin{align*}
\|r\partial_r(\tilde{g}_\sigma \ast \phi_L)\|_{L^2(\Omega)}^2 &\leq R_\Omega^2 \| \tilde{g}_\sigma \ast  \partial_r \phi_L\|_{L^2(\Omega)}^2 \\
&\leq R_\Omega^2 \|\partial_r\phi_L\|_1^2  \| \tilde{g}_\sigma\|_{L^2(\Omega)}^2.
\end{align*} Since $\phi$ is a radial filter and 
$$\partial_r \phi_L(r) = -(2\pi L^3)^{-1/2} re^{-r^2/(2L^2)}\, ,$$
we obtain
$$\|\partial_r\phi_L\|_1 = \frac{1}{L^3} \int_{0}^\infty  re^{-r^2/(2L^2)} \, dr = L^{-1}$$
so that $\|\partial_r\phi_L\|_1^2 = L^{-2}.$ Combining with the bound for $\|\tilde{g}_\sigma\|_{L^2(\Omega)}^2$, we obtain 
$$\|r\partial_r(\tilde{g}_\sigma \ast \phi_L)\|_{L^2(\Omega)}^2 \lesssim_{\Omega, f}  L^{-2}\left(\frac{\sigma^2 \lor \sigma^6}{M}\right).$$
Combining terms finishes the proof.
\end{proof} 

\begin{remark}
\label{rmk:conv_rate}
When $\sigma \geq 1$, the upper bound in Theorem \ref{thm:main} is minimized by choosing $L =  \frac{\sigma}{M^{1/6}}$, 
yielding a convergence rate of $O(\frac{\eta^2}{M} + \frac{\sigma^4}{M^{2/3}})$ on the squared error. 
\end{remark}

\section{Numerical Implementation of Bispectrum Recovery}
\label{sec:optimization}

Theorem \ref{thm:main} establishes that the estimator $\widetilde{Bf}$ defined in \eqref{Model 2 estimator} is an unbiased estimator of $Bf$ and quantifies the rate of convergence. Note the estimator requires applying the operator $(I-L_{C_0})^{-1}C_1L_{C_2}$ to the empirical data term
\begin{equation}
\label{equ:emp_data_term}
 \tilde{d} :=  4(\tilde{g}_\eta + \tilde{g}_\sigma) \ast \phi_L + r(\tilde{g}_\eta + \tilde{g}_\sigma) \ast \partial_r\phi_L \, ,   
\end{equation}
and one does not have a closed form formula for applying this operator. However, $\widetilde{Bf}$ is the unique solution of a quadratic programming problem, namely:
\begin{align}
\label{equ:convex_opt_Bf}
 \widetilde{Bf} &= \text{argmin}_{\mathring{g}} \mathfrak{ L}(\mathring{g}) \ , \\ \mathfrak{L}(\mathring{g}) &=  \left\|(I-L_{C_0})\mathring{g} - C_1L_{C_2}\tilde{d} \right\|_2^2 \, . \notag
\end{align} 
It is thus straight-forward to compute $\widetilde{Bf}$ by minimizing the convex function $\mathfrak{L}(\mathring{g})$. 
\begin{remark}
A computation similar to \cite{hirn2023power} shows the gradient of the loss is
$$\nabla_{\mathring{g}}\mathfrak{L}(\mathring{g}) = 2(I-L_{C_0})^{*}((I-L_{C_0})\mathring{g} - C_1L_{C_2}\tilde{d}) \, ,$$  
where the adjoint acts as
\begin{align*}
&(I-L_{C_0})^{*}h(\omega_1,\omega_2) \\
&\quad =  h(\omega_1, \omega_2) - C_0^2h(C_0^{-1}w_1,C_0^{-1}w_2)) \, .
\end{align*}
\end{remark}

There is however a caveat to the above: solving \eqref{equ:convex_opt_Bf} requires knowledge of $\eta$, since the constants $C_i$ depend on $\eta$, and this parameter is generally unknown. We thus first learn $\eta$ using the algorithm proposed in \cite{hirn2023power}, which simultaneously approximates $\eta$ and the power spectrum $Pf$ by solving the joint optimization:
\begin{align}
\label{equ:joint_opt}
(\widetilde{Pf}, \tilde{\eta} ) &= \text{argmin}_{\mathring{p} \geq 0 , \mathring{\eta} \geq 0} \mathcal{L}(\mathring{p},\mathring{\eta}) \, ,
\end{align}
where
\begin{align} 
\label{PSLossFunctionUnknownEta}
\mathcal{L}(\mathring{p},\mathring{\eta}) &= \frac{\| \left(I-L_{C_0(\mathring{\eta})}\right)\mathring{p} - C_1(\mathring{\eta})L_{C_2(\mathring{\eta})} \tilde{q} \|_2^2}{\mathring{\eta}^2} \, . 
\end{align}
Here $\tilde{q}$ is an empirical data term analogous to \eqref{equ:emp_data_term} but computed from the noisy \textit{power spectra}, and the definitions of $C_i(\eta)$ and $L_C$ have been slightly modified to achieve power spectrum recovery (see \cite{hirn2023power} for exact definitions). We thus conveniently learn the power spectrum in addition to $\eta$, as the algorithm we apply in Section \ref{sec:numerics_HSrecovery} to recover the hidden signal requires knowledge of both the bispectrum and power spectrum. Unfortunately, unlike \eqref{equ:convex_opt_Bf}, the optimization \eqref{equ:joint_opt} is nonconvex.

\begin{remark}
  The objective function in \cite{hirn2023power} in fact contains only the numerator of \eqref{PSLossFunctionUnknownEta}, but we found the modified loss \eqref{PSLossFunctionUnknownEta} yielded more stable recovery of $\eta$, since the numerator vanishes whenever $\mathring{\eta}=0$, 
  regardless of $\mathring{p}$.  
\end{remark} 

\begin{remark}
    Although the theoretical guarantees of Theorem \ref{thm:main} apply to continuous signals $f$, in practice we will only observe a discretization of these continuous signals, and the optimization in \eqref{equ:convex_opt_Bf} is not carried out on a function $\mathring{g}:\mathbb{R}^2 \rightarrow \mathbb{C}$, but on a matrix $\mathring{g}_{i,j}=\mathring{g}(\omega_i,\omega_j)$ of discrete frequencies. Thus in addition to the sampling error reported in Theorem \ref{thm:main}, we will in practice also be subjected to discretization errors arising from approximating the continuum Fourier transform with the DFT, approximating derivatives using difference quotients, and approximating dilation operators using interpolation. See \cite{hirn2023power} for a more detailed discussion of these discretization effects.
\end{remark}


\section{Numerical Experiments for Bispectrum Recovery}
\label{sec:numerics_BSrecovery}

\begin{figure}
	\centering
	\begin{subfigure}[b]{0.24\textwidth}
		\centering
		\includegraphics[width=\textwidth]{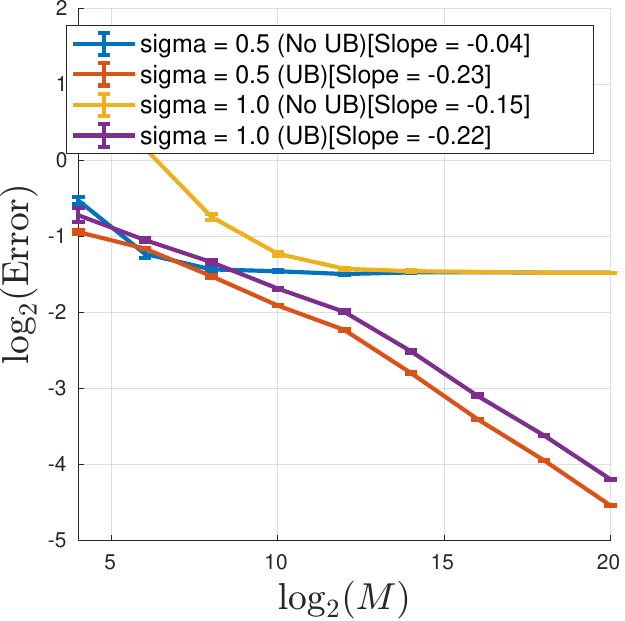}
		\caption{\footnotesize{$f_1$ (slopes $=-0.23, -0.22$)}}
		\vspace*{.3cm}
	\end{subfigure}
	\hfill
	\begin{subfigure}[b]{0.24\textwidth}
		\centering
		\includegraphics[width=\textwidth]{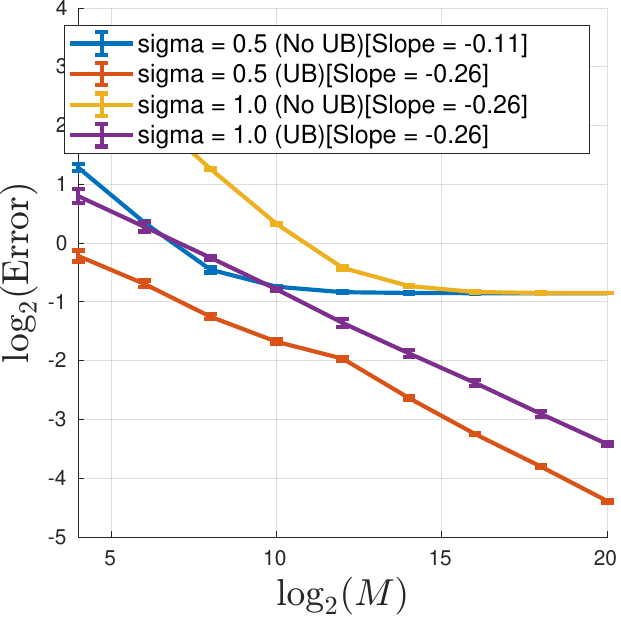}
		\caption{\footnotesize{$f_2$ (slopes $=-0.26, -0.26$)}}
		\vspace*{.3cm}
	\end{subfigure}
	\hfill
	\begin{subfigure}[b]{0.24\textwidth}
		\centering
		\includegraphics[width=\textwidth]{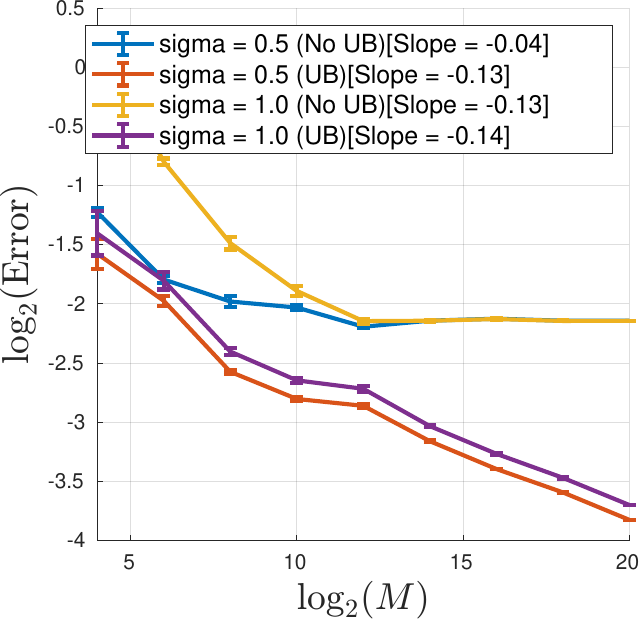}
		\caption{\footnotesize{$f_3$ (slopes $=-0.13, -0.14$)}}
	\end{subfigure}
	\begin{subfigure}[b]{0.24\textwidth}
		\centering
		\includegraphics[width=\textwidth]{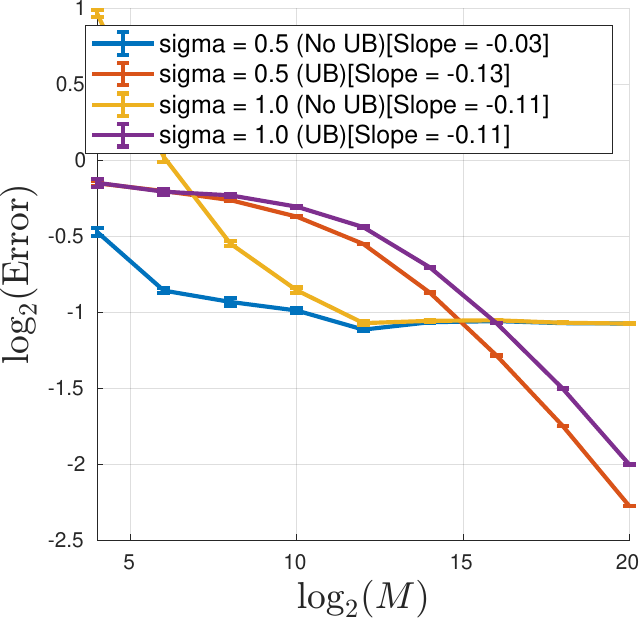}
		\caption{\footnotesize{$f_4$ (slopes $=-0.13, -0.11$)}}
	\end{subfigure}
	\caption{Bispectrum recovery error with oracle $\sigma,\eta$.}
	\label{bispectrum recovery oracle}
\end{figure}

\begin{figure}
	\centering
	\begin{subfigure}[b]{0.24\textwidth}
		\centering
		\includegraphics[width=\textwidth]{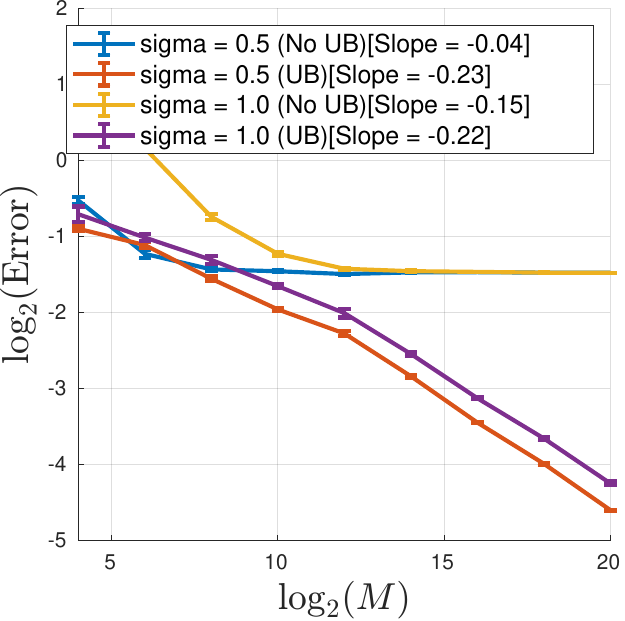}
		\caption{\footnotesize{$f_1$ (slopes $=-0.23, -0.22$)}}
		\vspace*{.3cm}
	\end{subfigure}
	\hfill
	\begin{subfigure}[b]{0.24\textwidth}
		\centering
		\includegraphics[width=\textwidth]{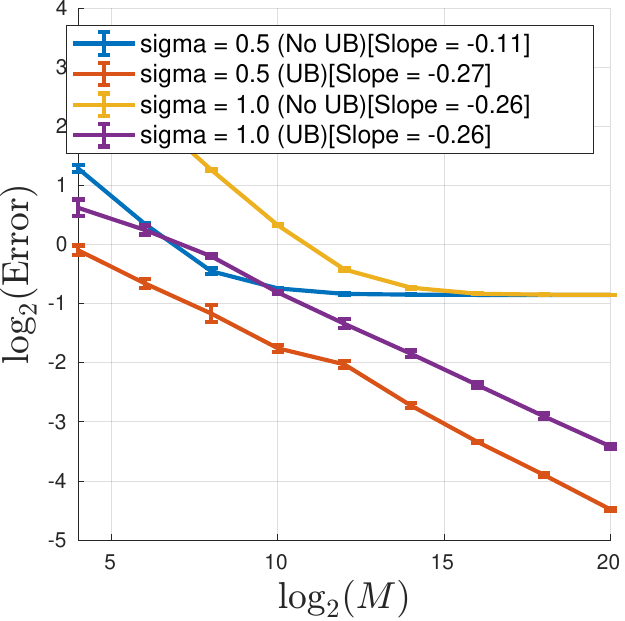}
		\caption{\footnotesize{$f_2$ (slopes $=-0.27, -0.26$)}}
		\vspace*{.3cm}
	\end{subfigure}
	\hfill
	\begin{subfigure}[b]{0.24\textwidth}
		\centering
		\includegraphics[width=\textwidth]{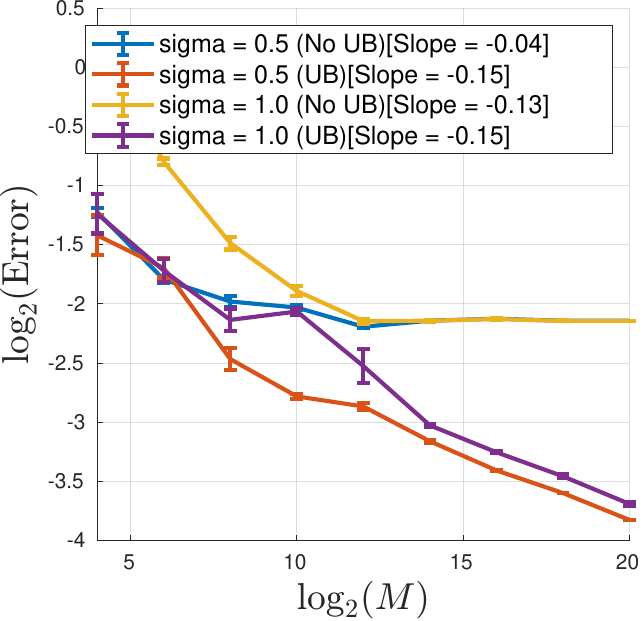}
		\caption{\footnotesize{$f_3$ (slopes $=-0.15, -0.15$)}}
	\end{subfigure}
	\begin{subfigure}[b]{0.24\textwidth}
		\centering
		\includegraphics[width=\textwidth]{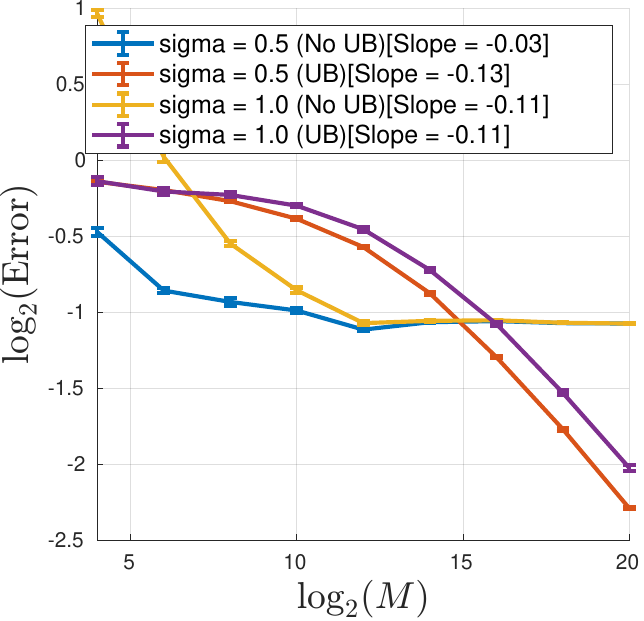}
		\caption{\footnotesize{$f_4$ (slopes $=-0.13, -0.11$)}}
	\end{subfigure}
	\caption{Bispectrum recovery error with unknown $\sigma,\eta$.}
\label{bispectrum recovery empirical}
\end{figure}

\begin{figure}[tbh]
\begin{minipage}[t]{0.320\columnwidth}
  \includegraphics[width=\linewidth]
{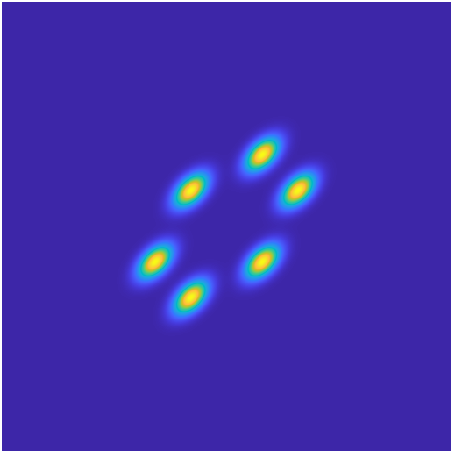}
\end{minipage}\hfill 
\begin{minipage}[t]{0.320\columnwidth}
\includegraphics[width=\linewidth]{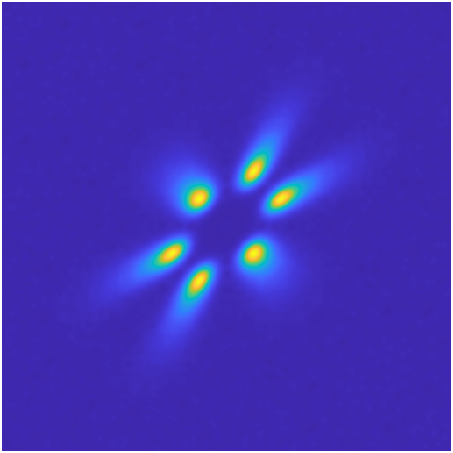}
\end{minipage}\hfill 
\begin{minipage}[t]{0.320\columnwidth}
\includegraphics[width=\linewidth]
{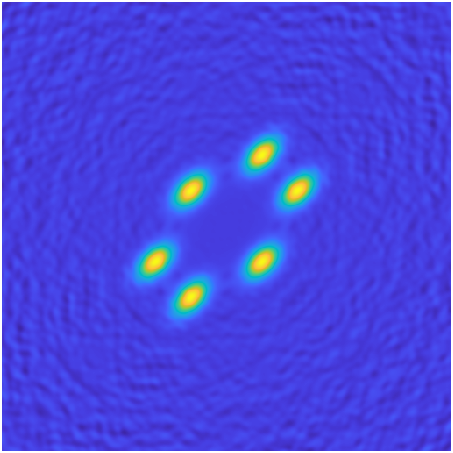}
\end{minipage}\hfill 
\caption{Example of bispectrum recovery for $f_2$ with $\eta = 12^{-1/2}$ and $\sigma = 1.0$. \textbf{Left:} ground truth signal. \textbf{Middle:} average of bispectra after additive noise unbiasing. \textbf{Right:} recovered bispectrum using unbiasing procedure.}
\label{bispectrum recovery example}
\end{figure}

\begin{figure}
	\centering
	\begin{subfigure}[b]{0.24\textwidth}
		\centering
		\includegraphics[width=\textwidth]{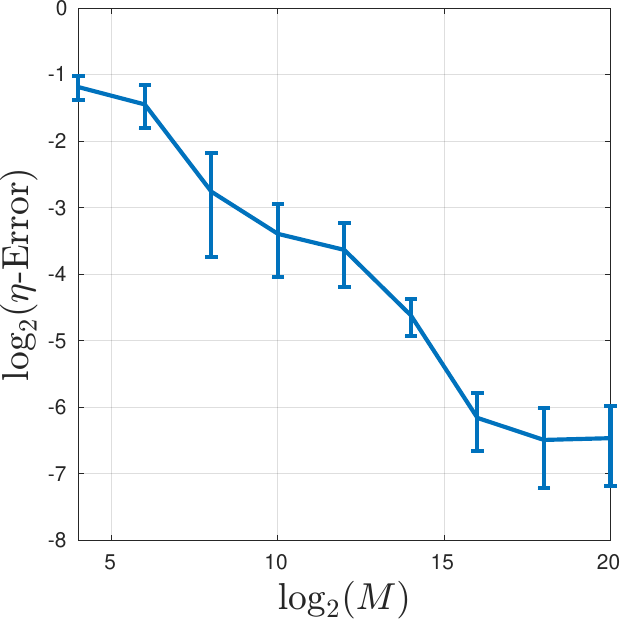}
		\caption{\footnotesize{$f_1$}}
		\vspace*{.3cm}
	\end{subfigure}
	\hfill
	\begin{subfigure}[b]{0.24\textwidth}
		\centering
		\includegraphics[width=\textwidth]{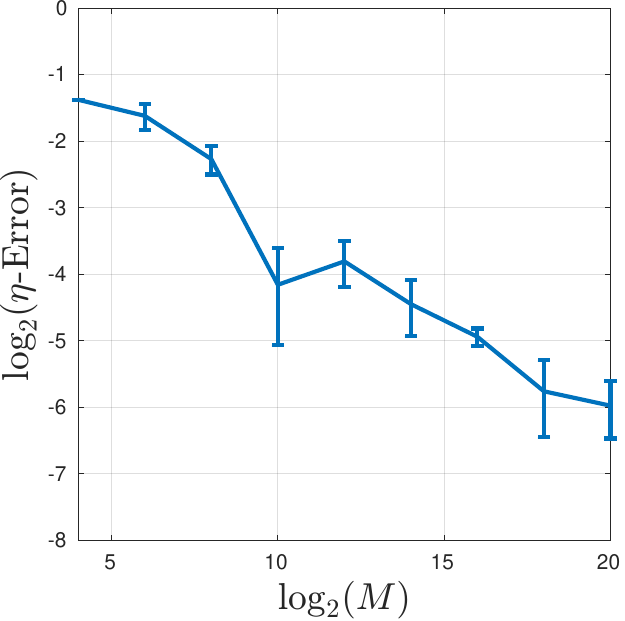}
		\caption{\footnotesize{$f_2$}}
		\vspace*{.3cm}
	\end{subfigure}
	\hfill
	\begin{subfigure}[b]{0.24\textwidth}
		\centering
		\includegraphics[width=\textwidth]{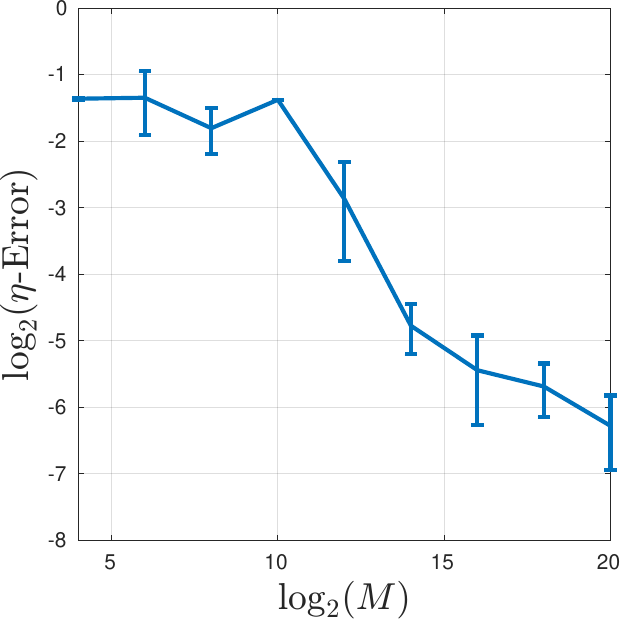}
		\caption{\footnotesize{$f_3$}}
	\end{subfigure}
	\begin{subfigure}[b]{0.24\textwidth}
		\centering
		\includegraphics[width=\textwidth]{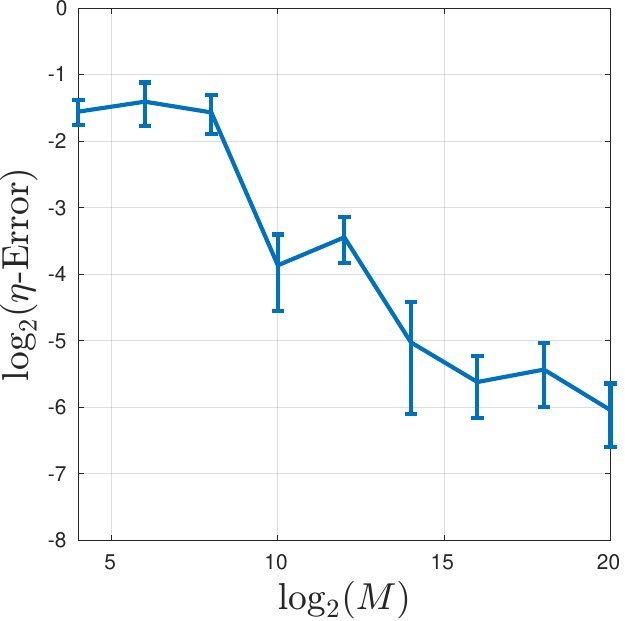}
		\caption{\footnotesize{$f_4$}}
	\end{subfigure}
	\caption{$\eta$ estimation plots with standard error bars.}
\label{fig: eta estimation}
\end{figure}

We will now test the accuracy of our bispectrum estimator $\widetilde{Bf}$ given in \eqref{Model 2 estimator} on the following signals:
\begin{align*}
f_1(x) &= A_1 e^{-5x^2} \cos(8x) \\
f_2(x) &= A_2 e^{-5x^2} \cos(12x) \\
f_3(x) &= A_3 \text{sinc}(4x) \\
f_4(x) &= A_4 \cos(6 x) \mathbf{1}_{\{x\in (-\frac{\pi}{4},\frac{\pi}{4})\}}.
\end{align*}
Like in \cite{hirn2021wavelet}, we define our hidden signals on $[-N/4, N/4]$ and the noisy signals on $[-N/2, N/2]$, and then sample them at a rate of $2^{-\ell}$, choosing $N = 2^5$ and $\ell=4$. In frequency, the signals were sampled on the interval $[-2^\ell \pi, 2^\ell \pi]$ with sampling rate of $\pi/N.$ The constants $A_i$ with $i = 1,\ldots,4$ were chosen so that the SNR for each $f_i$ was set to be $\sigma^{-2}$, where $\text{SNR} = \left(\int_{-N/2}^{N/2} |f(x)|^2 \, dx\right)/\sigma^2.$ 
Signals were randomly translated and randomly dilated with $\tau$ uniform on $[-\frac{1}{2}, \frac{1}{2}]$, so that the dilation factor $(1-\tau)\in [\frac{1}{2},2]$ and $\eta^2 = \frac{1}{12}$. Signals were then corrupted by Gaussian noise at two different noise levels: $\sigma=0.5$ and $\sigma=1$.

These signals were chosen to test the robustness of our proposed method. The signals $f_1$ and $f_2$ are smooth with exponential decay in both space and frequency, and thus the proposed method is expected to work well. 
On the other hand, $f_3$ is smooth in space, but discontinuous in frequency, so in this example the assumption of a smooth bispectrum is violated. Also, $f_4$ is continuous but not smooth in space, 
and thus poses a challenge due to its slowly decaying FT.

We start with the case where $\sigma$ and $\eta$ are given by an oracle. Figure \ref{bispectrum recovery oracle} shows the relative $L^2$ error of bispectrum recovery, i.e. $\text{Error} = \| Bf - \widetilde{Bf}\|_2/\|Bf\|_2$, as the sample size $M$ is increased. NO UB is an average of the noisy bispectra with an empirical centering to account for the bias from additive noise, and UB is with our unbiasing procedure. Figure \ref{bispectrum recovery example} illustrates how the proposed unbiasing procedure is able to remove the effect of the dilations for a fixed signal ($f_2$) and sample size.
 All experiments were run 
 with a Gaussian width of $L=5\sigma{M}^{-1/6}$ (see Remark \ref{rmk:conv_rate}). Note based on Remark 1, we expect to see the error decay approximately like $O(M^{-1/3})$, i.e. we expect to observe a slope of $-1/3$ on a log-log plot. For $f_1$ and $f_2$, the error decay is close to the theoretical bound (observed slope around -1/4). For $f_3$ and $f_4$ where some Model assumptions are violated, the decay is slower but still consistently linear, which yields empirical evidence that $\widetilde{Bf}$ is still an unbiased estimator of $Bf$. 

We next investigate bispectrum recovery when $\sigma$ and $\eta$ are unknown. Note when $\hat{f}$ has a fast decay, $\sigma$ can be accurately estimated from the decay of the power spectrum, i.e. by
\begin{align*}
	\widetilde{\sigma}^2 &:= \frac{1}{ \lvert \{ \omega_i\in\Sigma \} \rvert} \sum_{\omega_i \in \Sigma} \lvert \widehat{y}(\omega_i)\rvert^2 \, ,
\end{align*} 
where $\Sigma = [-2^\ell\pi, 2^\ell \pi] \setminus [-2^{\ell-1}\pi, 2^{\ell-1} \pi]$.
Estimating $\eta$ is more difficult but can be accomplished via a joint optimization as described in Section \ref{sec:optimization}. Figure \ref{fig: eta estimation} shows the error for $\eta$ recovery, and we see that $\eta$ can be learned reliably for large $M$. 
Figure \ref{bispectrum recovery empirical} shows the bispectrum recovery error when $\eta,\sigma$ are empirically estimated, and the results are nearly identical to the oracle case given in Figure \ref{bispectrum recovery oracle}, indicating that accurate bispectrum estimation is still possible.

\section{Hidden Signal Recovery}
\label{sec:numerics_HSrecovery}

\begin{figure}
	\centering
	\begin{subfigure}[b]{0.24\textwidth}
		\centering
		\includegraphics[width=\textwidth]{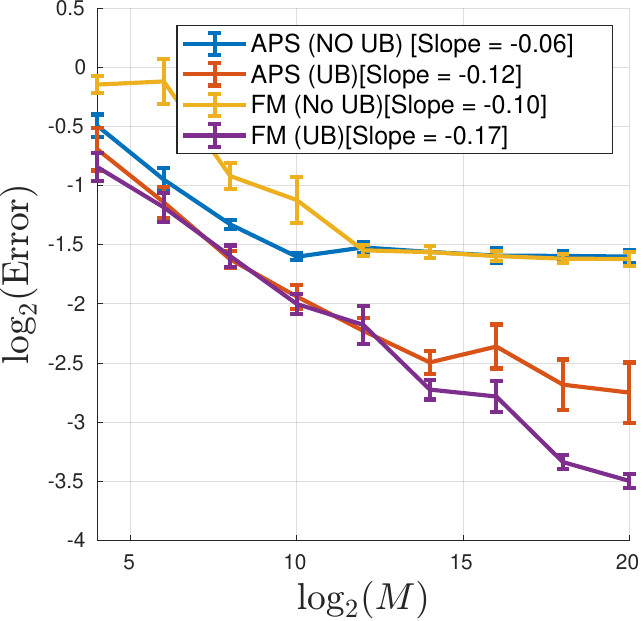}
		\caption{\footnotesize{$f_1$ (slopes$=-0.12,-0.17$)}}
		\vspace*{.3cm}
	\end{subfigure}
	\hfill
	\begin{subfigure}[b]{0.24\textwidth}
		\centering
		\includegraphics[width=\textwidth]{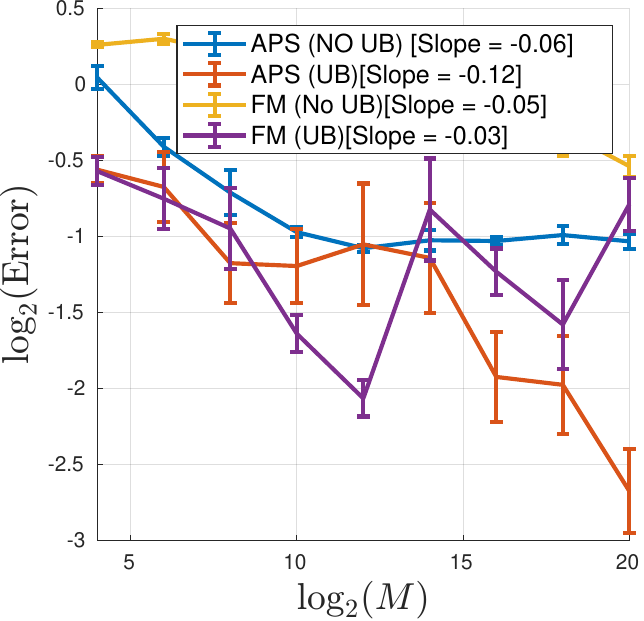}
		\caption{\footnotesize{$f_2$ (slopes$=-0.12,-0.03$)}}
		\vspace*{.3cm}
	\end{subfigure}
	\hfill
	\begin{subfigure}[b]{0.24\textwidth}
		\centering
		\includegraphics[width=\textwidth]{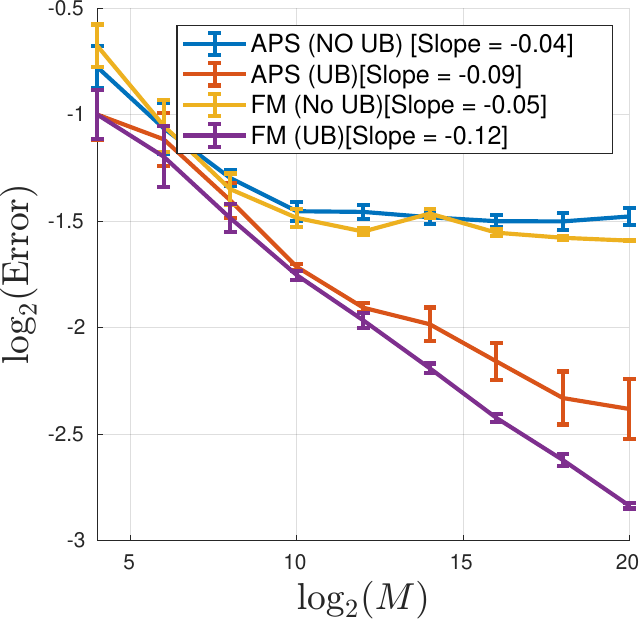}
		\caption{\footnotesize{$f_3$ (slopes$=-0.09,-0.12$)}}
	\end{subfigure}
	\begin{subfigure}[b]{0.24\textwidth}
		\centering
		\includegraphics[width=\textwidth]{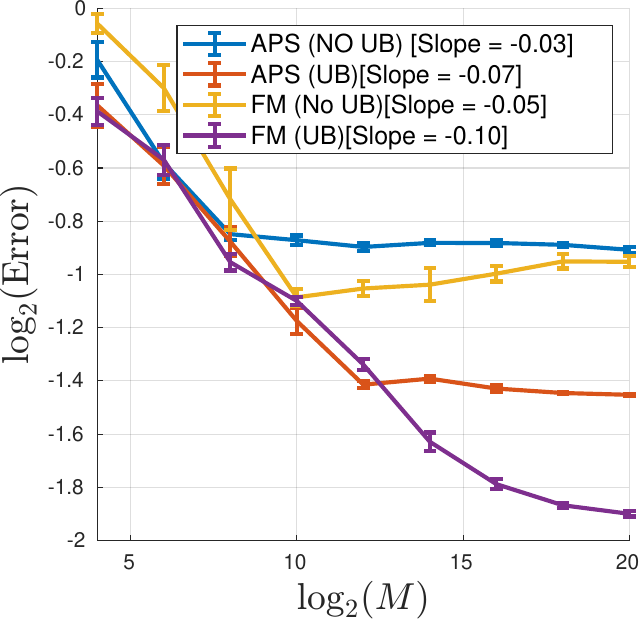}
		\caption{\footnotesize{$f_4$ (slopes$=-0.07,-0.10$)}}
	\end{subfigure}
	\caption{Hidden signal recovery error plots with oracle $\eta = 12^{-1/2}$ and $\sigma = 0.5$. First reported slope is for UB with APS, and second is UB with FM.}
\label{bispectrum inversion oracle}
\end{figure}

\begin{figure}
	\centering
	\begin{subfigure}[b]{0.24\textwidth}
		\centering
		\includegraphics[width=\textwidth]{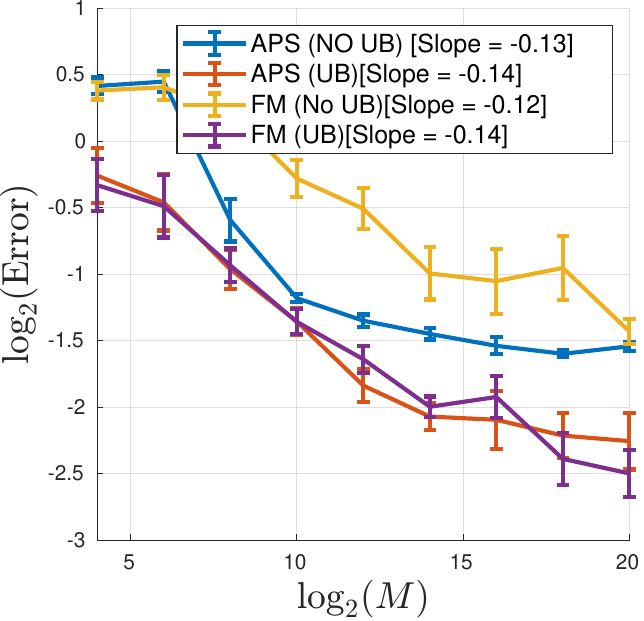}
		\caption{\footnotesize{$f_1$ (slopes$=-0.14,-0.14$)}}
		\vspace*{.3cm}
	\end{subfigure}
	\hfill
	\begin{subfigure}[b]{0.24\textwidth}
		\centering
		\includegraphics[width=\textwidth]{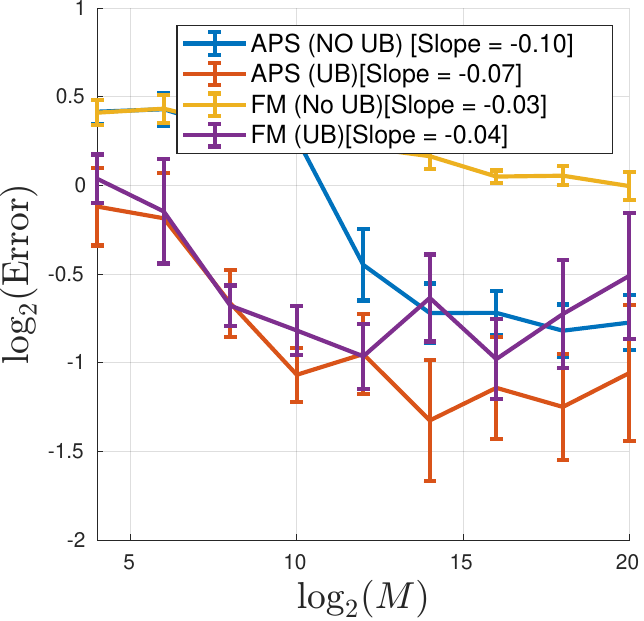}
		\caption{\footnotesize{$f_2$ (slopes$=-0.07,-0.04$)}}
		\vspace*{.3cm}
	\end{subfigure}
	\hfill
	\begin{subfigure}[b]{0.24\textwidth}
		\centering
		\includegraphics[width=\textwidth]{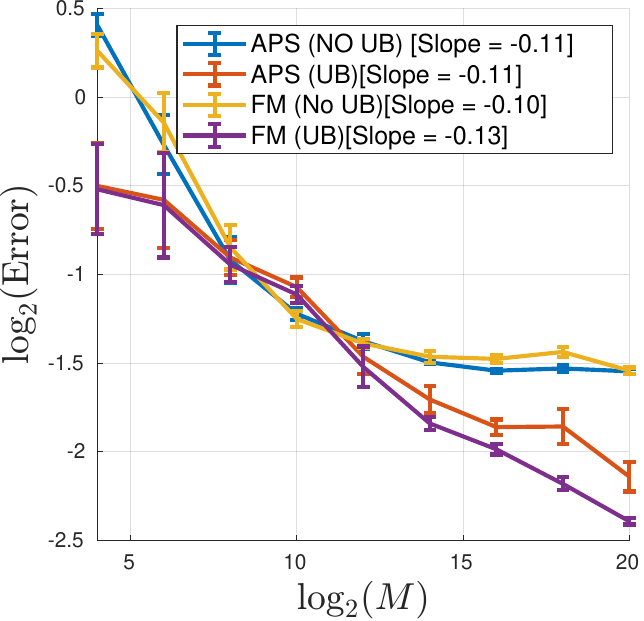}
		\caption{\footnotesize{$f_3$ (slopes$=-0.11,-0.13$)}}
	\end{subfigure}
	\begin{subfigure}[b]{0.24\textwidth}
		\centering
		\includegraphics[width=\textwidth]{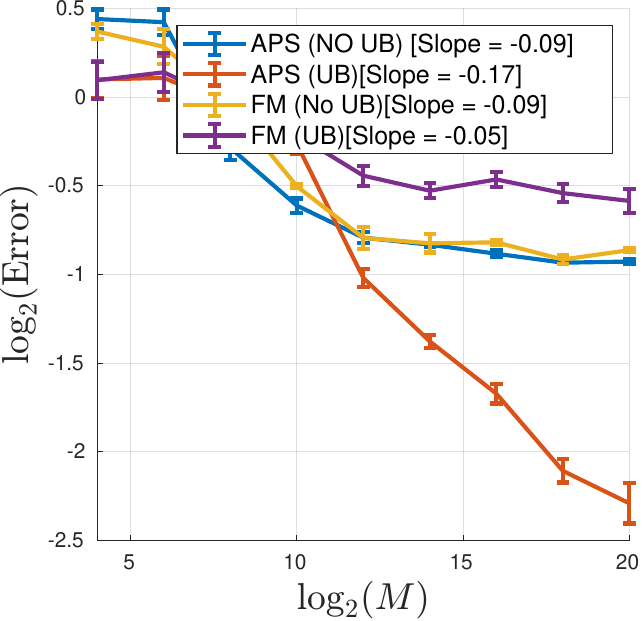}
		\caption{\footnotesize{$f_4$ (slopes$=-0.17,-0.05$)}}
	\end{subfigure}
	\caption{Hidden signal recovery error plot with unknown $\eta = 12^{-1/2}$ and $\sigma = 1.0$. First reported slope is for UB with APS, and second is UB with FM.}
\label{bispectrum inversion empirical}
\end{figure}

Once the bispectrum of the hidden signal is recovered, a bispectrum inversion algorithm must be applied to recover the hidden signal. Jennrich's algorithm provides a stability guarantee for this inversion \cite{perry2017sample, harshman1970foundations, leurgans1993decomposition}, but is not suitable for practical applications. Other approaches include non-convex optimization over the manifold of phases \cite{bendory2017bispectrum}, iterative phase synchronization \cite{bendory2017bispectrum}, semi-definite programming, phase unwrapping, frequency marching \cite{giannakis1989signal, sadler1992shift}, and the spectral method proposed in \cite{chen2018spectral}. 
We refer the reader to \cite{bendory2017bispectrum} for a survey and comparison of algorithms. Although we do not rigorously analyze inversion in this article, we nonetheless investigate empirically whether the bispectra recovered in Section \ref{sec:numerics_BSrecovery} can be accurately inverted to recover the hidden signal, i.e. whether the proposed method can fully solve the MRA problem. 
We invert our recovered bispectra using iterative phase synchronization (APS) and frequency marching (FM) (see \cite{bendory2017bispectrum}); note the inversion also requires knowledge of the power spectrum, which we learn using \cite{hirn2023power} as described in Section \ref{sec:optimization}. The resulting $L^2$ relative error of hidden signal recovery for the oracle and empirical cases are given in Figures \ref{bispectrum inversion oracle} and \ref{bispectrum inversion empirical} respectively. Note that the estimated signal $\tilde{f}$ is only defined up to translation, and the error is thus computed by $\min_t \| f(x) - \tilde{f}(x-t)\|_2 / \|f\|_2$. For brevity we only report the easiest case ($\sigma=0.5$, oracle) and the hardest case ($\sigma=1$, empirical). 
Although FM outperformed APS on the low frequency signals in the low noise regime, we found APS to be more reliable across a range of signals and noise levels (for example, FM performs poorly on the medium frequency Gabor $f_2$; see Figure \ref{bispectrum inversion oracle}). In general we observe that the hidden signal can be accurately estimated for large values of $M$, but the inversion process unfortunately increases the error, i.e. for large $M$ the hidden signal recovery error is typically 3-4 times larger than the bispectrum recovery error, and generally more variable. 


\begin{figure}[tbh]
	\centering
	\begin{subfigure}[b]{0.24\textwidth}
		\centering
		\includegraphics[width=\textwidth]{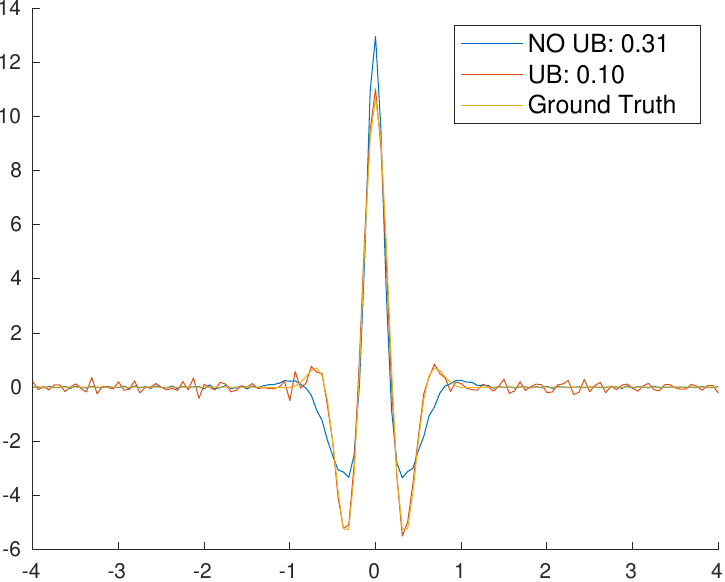}
		\caption{\footnotesize{$f_1$ (errors$=0.10,0.31$)}}
		\vspace*{.3cm}
	\end{subfigure}
	\hfill
	\begin{subfigure}[b]{0.24\textwidth}
		\centering
		\includegraphics[width=\textwidth]{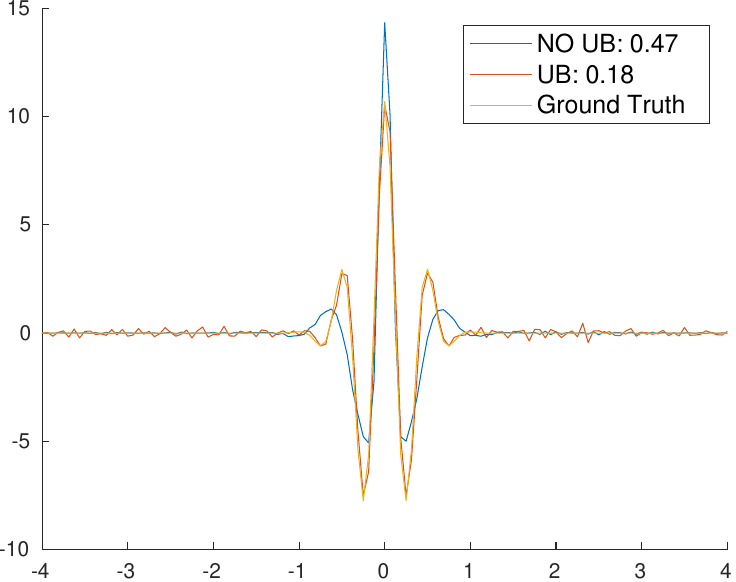}
		\caption{\footnotesize{$f_2$ (errors$=0.18,0.47$)}}
		\vspace*{.3cm}
	\end{subfigure}
	\hfill
	\begin{subfigure}[b]{0.24\textwidth}
		\centering
		\includegraphics[width=\textwidth]{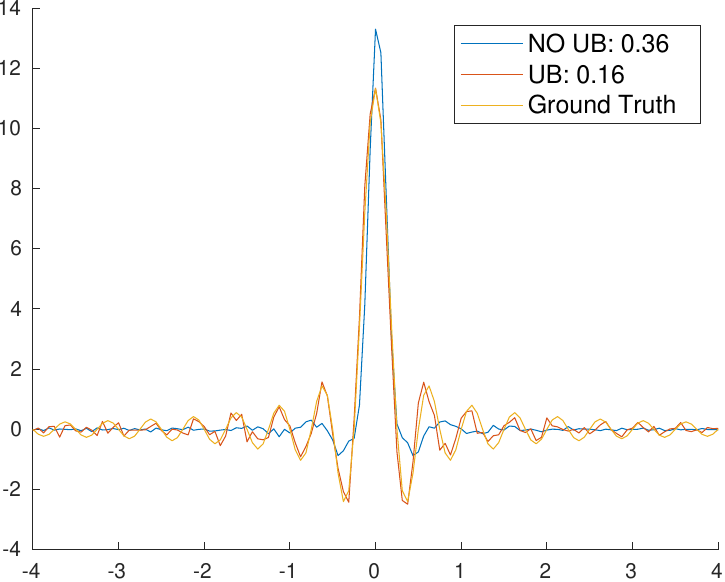}
		\caption{\footnotesize{$f_3$ (errors$=0.16,0.36$)}}
	\end{subfigure}
	\begin{subfigure}[b]{0.24\textwidth}
		\centering
		\includegraphics[width=\textwidth]{Pictures/f5_inv_plot_oracle_sigma_half.pdf}
		\caption{\footnotesize{$f_4$ (errors$=0.16,0.36$)}}
	\end{subfigure}
	\caption{Hidden signal recovery with oracle $\eta = 12^{-1/2}$ and $\sigma = 0.5$. 
 First reported error is for UB; second reported error is for NO UB.}
\label{bispectrum inversion example oracle}
\end{figure}

\begin{figure}[tbh]
	\centering
	\begin{subfigure}[b]{0.24\textwidth}
		\centering
		\includegraphics[width=\textwidth]{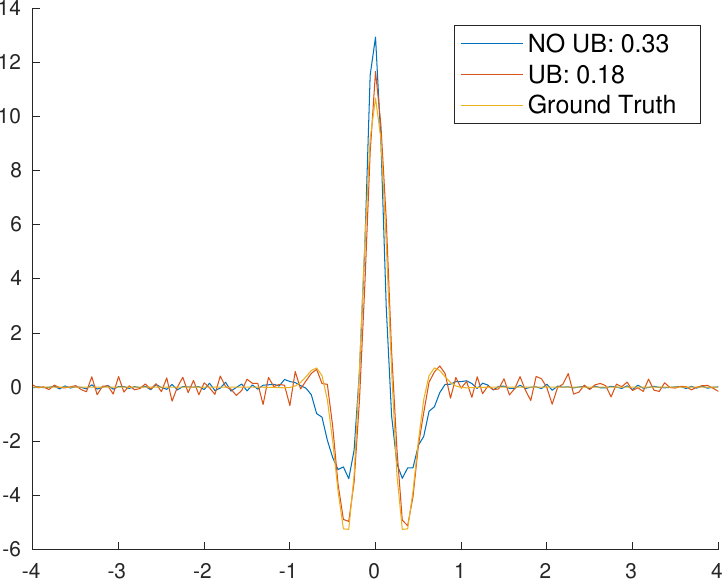}
		\caption{\footnotesize{$f_1$ (errors$=0.18,0.33$)}}
		\vspace*{.3cm}
	\end{subfigure}
	\hfill
	\begin{subfigure}[b]{0.24\textwidth}
		\centering
		\includegraphics[width=\textwidth]{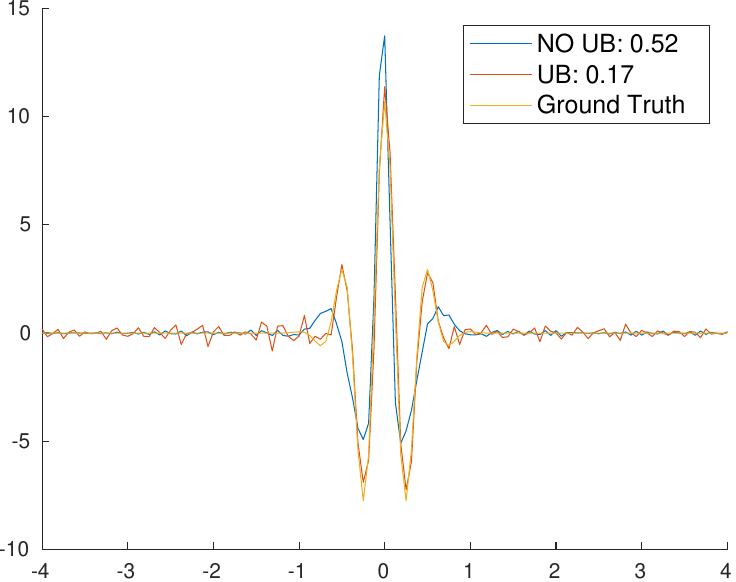}
		\caption{\footnotesize{$f_2$ (errors$=0.17,0.52$)}}
		\vspace*{.3cm}
	\end{subfigure}
	\hfill
	\begin{subfigure}[b]{0.24\textwidth}
		\centering
		\includegraphics[width=\textwidth]{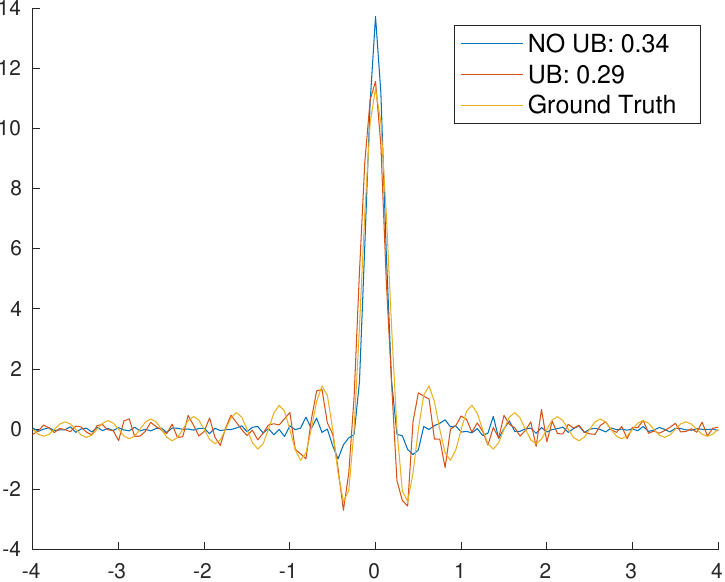}
		\caption{\footnotesize{$f_3$ (errors$=0.29,0.34$)}}
	\end{subfigure}
	\begin{subfigure}[b]{0.24\textwidth}
		\centering
		\includegraphics[width=\textwidth]{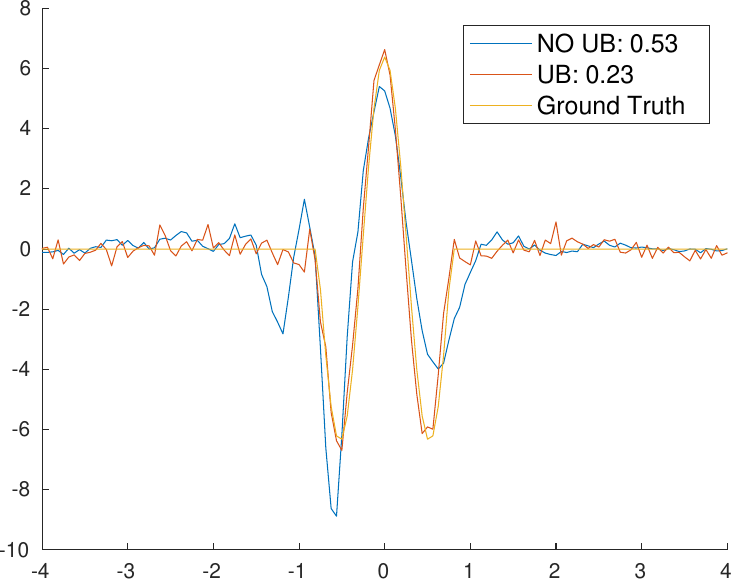}
		\caption{\footnotesize{$f_4$ (errors$=0.23,0.53$)}}
	\end{subfigure}
	\caption{Hidden signal recovery plots with unknown $\eta = 12^{-1/2}$ and $\sigma = 1.0$. 
    First reported error is for UB; second reported error is for NO UB.}
\label{bispectrum inversion example empirical}
\end{figure}

To complement our quantitative error plots, we also provide a qualitative assessment. In Figures \ref{bispectrum inversion example oracle} and \ref{bispectrum inversion example empirical}, examples of the output are provided using $M = 2^{20}$ samples (red curves), and we compare with the output if dilation unbiasing were not applied (blue curves). Our recovered signals capture well the general shape of the ground truth signals, and generally decrease the relative error of the hidden signal recovery by a factor of 2-4 (except for $f_3$ in the high noise regime). 
On the other hand, without our unbiasing procedure, the output contains bias from the dilatons, and the general shape of the recovered signals is incorrect, with the recovered peaks out of alignment. 


\section{Conclusion}
\label{sec:conclusion}
We study the MRA problem under the action of the non-compact dilation group, in addition to translation and additive noise. We present a novel technique for removing the dilation bias from the bispectrum, so that when the method is combined with bispectrum inversion it fully solves this MRA model. Since in many applications the assumption of a uniform distribution may fail, an important future direction is investigating whether the dilation distribution can be learned empirically and, if so, extending the method to an arbitrary but known dilation distribution. Also of interest is extending the method to higher dimensional signals, which is straight-forward for Model \ref{model2}, although computationally more burdensome. In higher dimensions however it is natural to incorporate rotations into the model, which poses additional challenges, as the rotation group in $\mathbb{R}^3$ is not commutative.

\bibliographystyle{plain} 
\bibliography{bibliography} 

\appendix

\subsection{Proof of Theorem 2}
\label{app:Thm2proof}
First, assume that the $Bf: \mathbb{R}^2 \to \mathbb{R}$. The argument will be generalized to the complex case after. Notice that 
$$|\tilde{g}_\eta(r, \theta) - g _\eta(r, \theta)|^2 = \left|\frac{1}{M} \sum_{j = 1}^M (Bf_j)(r, \theta) - g_\eta(r, \theta)\right|^2.$$
Define 
$$X_j  = Bf_j(r, \theta) - g_\eta(r, \theta) = Bf_j(r, \theta) - \mathbb{E}[(Bf_j)(r, \theta)].$$
This means each $X_j$ is zero centered, so we have
$$\mathbb{E}\left[\left|\frac{1}{M} \sum_{j=1}^M X_j\right|^2\right] = \text{var}\left[\frac{1}{M} \sum_{j=1}^M X_j\right] = \frac{\text{var}(X_j)}{M}.$$

Write
$$X_j = Bf_j(r, \theta) - (Bf)(r, \theta) + (Bf)(r, \theta) - \mathbb{E}[(Bf_j)(r, \theta)].$$
Then
\begin{align*}
X_j^2 &\lesssim (Bf_j(r, \theta) - (Bf)(r, \theta))^2 \\
&+ ((Bf)(r, \theta) - \mathbb{E}[(Bf_j)(r, \theta)])^2.
\end{align*}
and 
\begin{align*}
\mathbb{E}[X_j^2] &\lesssim \mathbb{E}\left[(Bf_j(r, \theta) - (Bf)(r, \theta))^2\right] \\
&+ \mathbb{E}\left[((Bf)(r, \theta) - \mathbb{E}[(Bf_j)(r, \theta)])^2\right]  \\
 &\lesssim \mathbb{E}\left[(Bf_j(r, \theta) - (Bf)(r, \theta))^2\right].
\end{align*} 

Each $\tau_j$ has bounded variance and is supported on $[-1/2,1/2]$. Taylor expand the dilated bispectrum in radial variable in  interval $[r/2, 2r]$:
\begin{align*}
(Bf)((1- \tau_j)r, \theta) &= (Bf)(r, \theta) - \partial_{r}(Bf)(r, \theta) r \tau_j \\
&+ \frac{1}{2} \partial_{rr}(Bf)(r, \theta)\bigg|_{r = \alpha} r^2 \tau_j^2,   
\end{align*}
with $\alpha \in [ r/2, 2r].$ Now multiply both sides by $(1- \tau_j)^3$ to get
\begin{align*}
(Bf_j)(r, \theta) &= (1- \tau_j)^3(Bf)((1- \tau_j)r, \theta) \\
                  &= (1- \tau_j)^3(Bf)(r, \theta) \\
                  &- (1- \tau_j)^3\partial_{r}(Bf)(r, \theta) r \tau_j \\
                  &+ (1- \tau_j)^3\frac{1}{2} \partial_{\alpha\alpha}(Bf)(\alpha, \theta) r^2 \tau_j^2.
\end{align*} with $\alpha \in [ r/2, 2r]$. It now follows that
\begin{align*}
&(Bf_j)(r, \theta) - (Bf)(r , \theta)\\
&= (3\tau_j^2 -3\tau_j - \tau_j^3)(Bf)(r, \theta) \\
&- (1- \tau_j)^3\partial_{r}(Bf)(r, \theta) r \tau_j \\
&+ \frac{1}{2} (1- \tau_j)^3\partial_{\alpha\alpha}(Bf)(\alpha, \theta) r^2 \tau_j^2.
\end{align*}
with $\alpha\in [ r/2, 2r].$

Square both sides to get:
\begin{align*}
&((Bf_j)(r, \theta) - (Bf)(r , \theta))^2 \\
&= (3\tau_j^2 - \tau_j - \tau_j^3)^2(Bf)^2(r, \theta) \\
&- 2(3\tau_j^2 - \tau_j - \tau_j^3)(1- \tau_j)^3(Bf)(r, \theta)\partial_{r}(Bf)(r, \theta) r \tau_j \\
&+(3\tau_j^2 - \tau_j - \tau_j^3)(1- \tau_j)^3(Bf)(r, \theta)\partial_{\alpha \alpha}(Bf)(\alpha, \theta) r^2 \tau_j^2 \\
&+ (1-\tau_j)^6[\partial_{r}(Bf)(r, \theta)]^2 r^2 \tau_j^2 \\
&-(1-\tau_j)^6\partial_{r}(Bf)(r, \theta)\partial_{\alpha \alpha}(Bf)(\alpha, \theta)r^3 \tau_j^3 \\
&+\frac{(1-\tau_j)^6}{4}[\partial_{\alpha \alpha}(Bf)(\alpha, \theta)]^2r^4 \tau_j^4.
\end{align*}
Using the inequality $2|ab| \leq |a|^2+|b|^2$, it follows that
\begin{align*}
&2|(\tau_j^2 - \tau_j - \tau_j^3)(1- \tau_j)^3(Bf)(r, \theta)\partial_{r}(Bf)(r, \theta) r \tau_j| \\
&\leq (\tau_j^2 - \tau_j - \tau_j^3)^2|(Bf)(r, \theta)|^2 \\
&+ (1- \tau_j)^6|\partial_r (Bf)(r, \theta) r \tau_j|^2
\end{align*}
and
\begin{align*}
&|(3\tau_j^2 - \tau_j - \tau_j^3)(1- \tau_j)^3(Bf)(r, \theta)\partial_{\alpha \alpha}(Bf)(\alpha, \theta) r^2 \tau_j^2| \\
& \leq \frac{1}{2}(3\tau_j^2 - \tau_j - \tau_j^3)^2 |(Bf)(r, \theta)|^2 \\
&+ \frac{1}{2}|1- \tau_j|^6 |\partial_{\alpha \alpha}(Bf)(\alpha, \theta) r^2 \tau_j^2|^2.
\end{align*}
and
\begin{align*}
&|(1-\tau_j)^6\partial_{r}(Bf)(r, \theta)\partial_{\alpha \alpha}(Bf)(\alpha, \theta)r^3 \tau_j^3|\\
& \leq \frac{1}{2}(1-\tau_j)^6 |\partial_{r}(Bf)(r, \theta)\tau_j r|^2 \\
&+ \frac{1}{2}|1- \tau_j|^6 |\partial_{\alpha \alpha}(Bf)(\alpha, \theta)r^2 \tau_j^2|^2.
\end{align*}

Take the expectation of both sides now. Since the pdf of $\tau$ is supported on $[-1/2, 1/2]$ and zero centered, 
$$\mathbb{E}[(3\tau_j^2 - \tau_j - \tau_j^3)^2] \lesssim \mathbb{E}[\tau_j^2] \lesssim \eta^2.$$
The other terms with $\tau_j$ are bounded in a similar way. Thus 
\begin{align*}
&\mathbb{E}[((Bf_j)(r, \theta) - (Bf)(r , \theta))^2] \\
&\lesssim \eta^2(Bf)^2(r, \theta) + r^2\eta^2[\partial_{r}(Bf)(r, \theta)]^2 \\
&+ \eta^4r^4 \left[\max_{\alpha \in[r/2,2r]}|\partial_{\alpha \alpha}(Bf)(\alpha, \theta)|\right]^2.
\end{align*}
We also have 
\begin{align*}
\text{var}[X_j] &= \mathbb{E}[X_j^2] \\
&\lesssim\mathbb{E}[((Bf_j)(r, \theta) - (Bf)(r , \theta))^2]\\
&\lesssim \eta^2(Bf)^2(r, \theta) + r^2\eta^2[\partial_{r}(Bf)(r, \theta)]^2 \\
&+ \eta^4 r^4\left[\max_{\alpha \in[r/2,2r]}|\partial_{\alpha \alpha}(Bf)(\alpha, \theta)|\right]^2,
\end{align*}
and 
\begin{align*}
&\mathbb{E}\left[(g_\eta(r, \theta) - \tilde{g}_\eta(r , \theta))^2\right] \\
&\lesssim  \frac{\eta^2}{M}(Bf)^2(r, \theta) + r^2\frac{\eta^2}{M}[\partial_{r}(Bf)(r, \theta)]^2 \\
&+ \frac{\eta^4}{M} r^4 \left[\max_{\alpha \in[r/2,2r]}|\partial_{\alpha \alpha}(Bf)(\alpha, \theta)|\right]^2.
\end{align*}
Now we can take the integral and expectation to get
\begin{align*}
&\mathbb{E}\left[\|g_\eta(r, \theta) - \tilde{g}_\eta(r , \theta)\|_2^2\right] \\
&\lesssim \frac{\eta^2}{M}\|(Bf)(r, \theta)\|_2^2 + \frac{\eta^2}{M} \|r \partial_{r}(Bf)(r, \theta)\|_2^2 \\
&+ \frac{\eta^4}{M} \left\|r^2 \max_{\alpha \in[r/2,2r]}|\partial_{\alpha \alpha}(Bf)(\alpha, \theta)|\right\|_2^2 
\end{align*}
The first term is now handled appropriately. We can now repeat a nearly identical argument for the second term. Let $g_j = Bf_j$ and 
$$Z_j = r \frac{\partial g_j}{\partial r} (r, \theta) - r\frac{\partial g_\eta}{\partial r} (r, \theta).$$
We have
$$r \frac{\partial \tilde{g}_\eta}{\partial r} (r, \theta) - r\frac{\partial g_\eta}{\partial r} (r, \theta) = \frac{1}{M}\sum_{j=1}^M r \frac{\partial g_j}{\partial r} (r, \theta) -  r\frac{\partial g_\eta}{\partial r} (r, \theta).$$
By Leibniz Rule, we can take the derivative inside the expectation to get $\mathbb{E}[Z_j] = 0$, and a similar argument from before yields
\begin{align*}
Z_j^2 &\lesssim \left[r\frac{\partial g_j}{\partial r}(r, \theta) - r\frac{\partial g}{\partial r}(r, \theta)\right]^2 \\
&+ \left[r\frac{\partial g}{\partial r}(r, \theta)- r\frac{\partial g_\eta}{\partial r}(r, \theta)\right]^2    
\end{align*}
and 
$$\mathbb{E}[Z_j^2] \lesssim \mathbb{E}\left[\left(r\frac{\partial g_j}{\partial r}(r, \theta) - r\frac{\partial g}{\partial r}(r, \theta)\right)^2\right].$$
Taylor expand $\partial_{r}(Bf)((1-\tau_j)r, \theta)$ to get 
\begin{align*}
\partial_{r}(Bf)((1- \tau_j)r, \theta) &= \partial_{r}(Bf)(r, \theta) - \partial_{rr}(Bf)(r, \theta) r \tau_j \\
&+ \frac{1}{2} \partial_{\gamma\gamma\gamma}(Bf)(\gamma, \theta) r^2 \tau_j^2,
\end{align*}
with $\gamma \in [ r/2, 2r].$ Since $r\frac{\partial g_j}{\partial r}(r, \theta) = r (1-\tau_j)^4 \partial_{r}(Bf)((1-\tau_j)r, \theta),$ multiply both sides by $(1- \tau_j)^4$:
\begin{align*}
&(1- \tau_j)^4r\partial_{r}(Bf)((1- \tau_j)r, \theta) \\
&= (1- \tau_j)^4r\partial_{r}(Bf)(r, \theta) \\
&- (1- \tau_j)^4\partial_{rr}(Bf)(r, \theta) r^2 \tau_j \\
&+ (1- \tau_j)^4\frac{1}{2} \partial_{\gamma\gamma\gamma}(Bf)(\gamma, \theta) r^3 \tau_j^2
\end{align*}
with $\gamma \in [ r/2, 2r].$ Then 
\begin{align*}
&r\frac{\partial g_j}{\partial r}(r, \theta) - r\frac{\partial g}{\partial r}(r, \theta) \\
&= (\tau_j^4 - 4\tau_j^3 + 6\tau_j^2 - 4\tau_j)r\partial_{r}(Bf)(r, \theta) \\
&- (1- \tau_j)^4\partial_{rr}(Bf)(r, \theta) r^2 \tau_j \\
&+ (1- \tau_j)^4\frac{1}{2} \partial_{\gamma\gamma\gamma}(Bf)(\gamma, \theta) r^3 \tau_j^2
\end{align*}
with $\gamma \in [ r/2, 2r].$ By a similar process from above, 
\begin{align*}
&\mathbb{E}\left[\left(r\frac{\partial g_j}{\partial r}(r, \theta) - r\frac{\partial g_\eta}{\partial r}(r, \theta)\right)^2\right]\\
&\lesssim \frac{\eta^2}{M}\|r\partial_{r}(Bf)(r, \theta)\|_2^2 + \frac{\eta^2}{M}\|r^2\partial_{rr}(Bf)(r, \theta)\|_2^2 \\
&+ \frac{\eta^4}{M} \left\|r^3 \max_{\gamma \in[r/2,2r]}|\partial_{\gamma\gamma\gamma}(Bf)(\gamma, \theta)|\right\|^2_2.
\end{align*}
Thus, we get the desired bound in the real case. For the case where $Bf: \mathbb{R}^2 \to \mathbb{C}$, simply write $Bf = \text{Re}(Bf) + i \text{Im}(Bf)$ and repeat the argument above on the real and imaginary parts and add the two norms together. 

\subsection{Proof of Lemma \ref{lem:noise_bound}}
\label{app:Lemproof}
We use triangle inequality to get
\begin{align*}
&\|\tilde{g}_\sigma\|_{L^2(\Omega)}^2 \\ 
& \lesssim \left\|\frac{1}{M}\sum_{j=1}^M\hat{f}_j(\omega_1)\hat{f}_j(\omega_2)\hat{\varepsilon}_j(\omega_2 - \omega_1)\right\|_{L^2(\Omega)}^2 \\
&+  \left\|\frac{1}{M}\sum_{j=1}^M\hat{f}_j^{*}(\omega_2)\hat{f}_j(\omega_2 - \omega_1)\hat{\varepsilon}_j(\omega_1)\right\|^2_{L^2(\Omega)}\\ 
&+  \left\|\frac{1}{M}\sum_{j=1}^M\hat{f}_j(\omega_1)\hat{f}_j(\omega_2 - \omega_1)\hat{\varepsilon}_j^{*}(\omega_2)\right\|^2_{L^2(\Omega)}\\
&+  \left\|\frac{1}{M}\sum_{j=1}^M B\varepsilon_j(\omega_1, \omega_2)\right\|_{L^2(\Omega)}^2\\
& +\left\|B_1\right\|^2_{L^2(\Omega)} + \left\|B_2\right\|^2_{L^2(\Omega)} +  \left\|B_3\right\|^2_{L^2(\Omega)}.
\end{align*} 
where
\begin{align*}
B_1 &= \frac{1}{M}\sum_{j=1}^M\hat{f}_j(\omega_1)\hat{\varepsilon}_j^{*}(\omega_2)\hat{\varepsilon}_j(\omega_2 - \omega_1) \\
&- h(\omega_1) \tilde{\mu}(\omega_1)\\
B_2 &= \frac{1}{M}\sum_{j=1}^M\hat{f}_j^{*}(\omega_2)\hat{\varepsilon}_j(\omega_1)\hat{\varepsilon}_j(\omega_2 - \omega_1) \\
&- h(\omega_2) \tilde{\mu}^{*}(\omega_2)\\
B_3 &= \frac{1}{M}\sum_{j=1}^M\hat{f}_j(\omega_2-\omega_1)\hat{\varepsilon}_j(\omega_1)\hat{\varepsilon}_j^{*}(\omega_2) \\
&- h(\omega_2-\omega_1) \tilde{\mu}(\omega_2 - \omega_1)
\end{align*}
and we proceed to bound the expectation of each term.
For the first term, we obtain:
\begin{small}
\begin{align*}
&\mathbb{E}\left[\left\|\frac{1}{M}\sum_{j=1}^M\hat{f}_j(\omega_1)\hat{f}_j^{*}(\omega_2)\hat{\varepsilon}_j(\omega_2 - \omega_1)\right\|_{L^2(\Omega)}^2\right] \\
&= \int_\Omega \mathbb{E}\left[ \frac{1}{M^2}\left|\sum_{j=1}^M\hat{f}_j(\omega_1)\hat{f}_j(\omega_2)\hat{\varepsilon}_j(\omega_2 - \omega_1)\right|^2\right] \, d \omega_1 \, d \omega_2 \\
&= \frac{\sigma^2}{M^2}\int_\Omega \sum_{j=1}^M|\hat{f}_j(\omega_1)\hat{f}_j(\omega_2)|^2  \, d \omega_1 \, d \omega_2 \\
&= \frac{\sigma^2}{M^2}\sum_{j=1}^M \int_\Omega|\hat{f}_j(\omega_1)|^2\, d \omega_1 \int_\Omega|\hat{f}_j(\omega_2)|^2 \, d \omega_2 \\
&= \frac{\sigma^2}{M} \|\hat{f}\|_{L^2(\Omega)}^4 \leq \frac{\sigma^2}{M} (2\pi)^2 \|f\|_{L^2(\mathbb{R})}^4 \, .
\end{align*} 
\end{small}
An identical argument can be applied to bound the second and third terms. For the fourth term, since $\mathbb{E}[B\epsilon_j]=0$, we have
\begin{align*}
&\mathbb{E}\left[\left\|\frac{1}{M}\sum_{j=1}^M B\varepsilon_j(\omega_1, \omega_2) \right\|_{L^2(\Omega)}^2\right] \\
&= \mathbb{E}\left[\int_{\Omega} \left|\frac{1}{M} \sum_{j=1}^M B\varepsilon_j(\omega_1,\omega_2)\right|^2 \, d \omega_1 \, d \omega_2\right] \\
&= \int_{\Omega} \mathbb{E}\left[\left|\frac{1}{M} \sum_{j=1}^M B\varepsilon_j(\omega_1,\omega_2)\right|^2\right]  \, d \omega_1 \, d \omega_2\\
&= \int_{\Omega} \text{Var}\left[\frac{1}{M} \sum_{j=1}^M B\varepsilon_j(\omega_1,\omega_2)\right]  \, d \omega_1 \, d \omega_2\\
&=  \frac{1}{M} \int_{\Omega} \text{Var}(B \varepsilon_j) \, d \omega_1\, d \omega_2.
\end{align*}
By Holder's inequality and Lemma \ref{lem: noise bound},
\begin{small}
\begin{align*}
&\text{Var}(B \varepsilon_j)
= \mathbb{E}\left[ |\hat{\varepsilon}(\omega_1)|^2| \hat{\varepsilon}(\omega_2)|^2 |\hat{\varepsilon}(\omega_2-\omega_1)|^2\right]\\
& \leq \mathbb{E}\left[|\hat{\varepsilon}(\omega_1)|^6\right]^{\frac{1}{3}} \mathbb{E}\left[| \hat{\varepsilon}(\omega_2)|^6\right]^{\frac{1}{3}} \mathbb{E}\left[|\hat{\varepsilon}(\omega_2-\omega_1)|^6\right]^{\frac{1}{3}} \lesssim \sigma^6 \, ,
\end{align*} 
\end{small}
and we thus obtain
$$\mathbb{E}\left[\left\|\frac{1}{M}\sum_{j=1}^M B\varepsilon_j(\omega_1, \omega_2) \right\|_{L^2(\Omega)}^2\right]\lesssim \frac{\sigma^6}{M} |\Omega|.$$

We now bound the last three terms, starting with 
$\mathbb{E}\left[\left\|B_1\right\|^2_{L^2(\Omega)}\right].$ 
Consider the random variable
$$A_j = \hat{f}_j(\omega_1)\hat{\varepsilon}_j^{*}(\omega_2)\hat{\varepsilon}_j(\omega_2 - \omega_1).$$
Then 
\begin{align*}
\mathbb{E}\left[\left\|B_1\right\|^2_{L^2(\Omega)}\right] 
& \lesssim \mathbb{E}\left[\left\|\frac{1}{M}\sum_{j=1}^M A_j - h(\omega_1)\mu(\omega_1) \right\|_{L^2(\Omega)}^2\right]  \\
&+ \mathbb{E}\left[\left\|h(\omega_1)\mu(\omega_1) - h(\omega_1) \tilde{\mu}(\omega_1)\right\|^2_{L^2(\Omega)}\right].
\end{align*}
Since $\Ex[A_j] = h(\omega_1)\mu(\omega_1)$,
\begin{small}
\begin{align*}
&\mathbb{E}\left[\left\|\frac{1}{M}\sum_{j=1}^M A_j -h(\omega_1)\mu(\omega_1)\right\|_{L^2(\Omega)}^2\right] 
=  \int_{\Omega} \frac{1}{M}\text{Var}\left[A_j\right] \, d \omega_1 \, d\omega_2 \, ,
\end{align*} 
\end{small}
and
\begin{align*}
    \text{Var}(A_j) & = \mathbb{E}[|\hat{f}_j(\omega_1)\hat{\varepsilon}_j^{*}(\omega_2)\hat{\varepsilon}_j(\omega_2 - \omega_1)|^2] -h^2(\omega_1)\mu^2(\omega_1) \\
    &\leq \mathbb{E}[|\hat{f}_j(\omega_1)\hat{\varepsilon}_j^{*}(\omega_2)\hat{\varepsilon}_j(\omega_2 - \omega_1)|^2]\\
    &\lesssim \sigma^4 \mathbb{E}_{t,\tau}[|\hat{f}_j(\omega_1)|^2]
\end{align*}
Now substitute this back into the integral to get 
\begin{align*}
\int_{\Omega} \frac{1}{M}\text{Var}\left[A_j\right] \, d \omega_1 \, d\omega_2 &\lesssim \frac{\sigma^4}{M}\int_{\Omega}\mathbb{E}_{t,\tau}[|\hat{f}_j(\omega_1)|^2] \, d \omega_1 \, d \omega_2\\
&=\frac{\sigma^4}{M}  \int_{\Omega}\mathbb{E}_\tau \left[|\hat{f}_j(\omega_1)|^2\right] \, d \omega_1 \, d\omega_2
\end{align*}
by translation invariance of the power spectrum. Now we have 
\begin{align*}
&\frac{\sigma^4}{M}  \int_{\Omega}\mathbb{E}_{\tau} \left[|\hat{f}_j(\omega_1)|^2\right] \, d \omega_1 \, d\omega_2 \\
&\qquad = \frac{\sigma^4}{M}\mathbb{E}_\tau\left[\int_{\Omega}     |\hat{f}_j(\omega_1)|^2\, d \omega_1 \, d\omega_2 \right]
\lesssim_{\Omega} \frac{\sigma^4}{M} \|f\|_2^2, 
\end{align*}
where the last line follows because $\|\hat{f}_j\|_2^2 = (1-\tau_j)\|\hat{f}\|_2^2 \leq \frac{3}{2} \|\hat{f}\|_2^2= 3\pi \|f\|_2^2$.

For the second term, 
\begin{align*}
&\mathbb{E}\left[\left\|h(\omega_1)\mu(\omega_1) - h(\omega_1)\tilde{\mu}(\omega_1)\right\|^2_{L^2(\Omega)}\right] \\
&= \int_{\Omega} h^2(\omega_1)\mathbb{E}\left[\left|\mu(\omega_1) - \tilde{\mu}(\omega_1)\right|^2\right] \, d\omega_1 \, d \omega_2 \\
&\lesssim\sigma^4 \int_{\Omega} \mathbb{E}\left[\left|\mu(\omega_1) - \tilde{\mu}(\omega_1)\right|^2\right] \, d\omega_1 \, d \omega_2 \\
& = \sigma^4 \int_{\Omega} \mathbb{E}\left[\left| \frac{1}{M}\sum_{j=1}^M \hat{y}_j(\omega_1) - \mu(\omega_1)\right|^2\right] \, d\omega_1 \, d \omega_2
\end{align*} 
Define $Z_j = \hat{y}_j(\omega) - \mu(\omega).$ Using similar steps to before, we get 
\begin{align*}
&\mathbb{E}\left[\left\|h(\omega_1)\mu(\omega_1) - h(\omega_1) \tilde{\mu}(\omega_1)\right\|^2_{L^2(\Omega)}\right] \\
& \lesssim \sigma^4 \int_{\Omega} \left(\frac{1}{M}\sum_{j=1}^M Z_j\right)^2 \, d\omega_1 \, d \omega_2 \\
&= \frac{\sigma^4}{M} \int_{\Omega} \text{Var}(Z_j) \, d\omega_1 \, d \omega_2 \\
& \lesssim_{\Omega} \frac{\sigma^4}{M} (\|f\|_{2}^2+\sigma^2).
\end{align*} 
Thus $\mathbb{E}\left[\left\|B_1\right\|^2_{L^2(\Omega)}\right] \lesssim_{\Omega} \frac{\sigma^4}{M} (\|f\|_{2}^2+\sigma^2)$ and an identical argument can be applied to control $B_2, B_3$.
Combining all terms thus gives
$$\|\tilde{g}_\sigma\|_{L^2(\Omega)}^2 \lesssim_{\Omega} \frac{\sigma^2}{M} \|f\|_2^4 + \frac{\sigma^4}{M} \|f\|_2^2 +  \frac{\sigma^6}{M}.$$

\end{document}